\documentclass[12pt]{amsart}
\usepackage[margin=1in]{geometry}
\usepackage[english]{babel}
\usepackage[utf8]{inputenc}
\usepackage{subcaption}
\usepackage{amsmath}
\usepackage{amssymb}
\usepackage{amsfonts}
\usepackage{amsthm}
\usepackage{mathrsfs}
\usepackage[all]{xy}
\usepackage[pdftex]{graphicx}
\usepackage{color}
\usepackage{url}
\usepackage{indentfirst}
\usepackage[labelfont=bf,labelsep=period,justification=raggedright]{caption}
\usepackage{hyperref}
\usepackage{algorithm}
\usepackage{algpseudocode}
\usepackage{tkz-fct}
\usepackage{tikz}
\usetikzlibrary{shapes.geometric}
\usetikzlibrary{calc}
\usepackage{multicol}
\PassOptionsToPackage{dvipsnames,svgnames}{xcolor}
\usepackage{textcomp}

\DeclareMathOperator{\dep}{dep}
\newcommand{\mcB}{\mathcal{B}}
\newcommand{\mcT}{\mathcal{T}}
\DeclareMathOperator{\card}{card}

\newcommand{\mcA}{\mathcal{A}}

\newcommand{\mcH}{\mathcal{H}}

\newcommand{\mcO}{\mathcal{O}}

\def\multiset#1#2{\ensuremath{\left(\kern-.3em\left(\genfrac{}{}{0pt}{}{#1}{#2}\right)\kern-.3em\right)}}

\theoremstyle{plain}
\newtheorem{thm}{Theorem}
\newtheorem{lemma}[thm]{Lemma}
\newtheorem{cor}[thm]{Corollary}

\newtheorem{prop}[thm]{Proposition}

\theoremstyle{definition}
\newtheorem{defn}[thm]{Definition}

\newtheorem*{notn}{Notation}
\theoremstyle{remark}
\newtheorem{rem}[thm]{Remark}
\newtheorem*{ex}{Example}

\numberwithin{equation}{section}

\tikzset{
  every picture/.style={line cap=round, line join=round},
  edge/.style={draw=black, line width=0.9pt},
  leaf/.style={circle, draw=black, fill=white,
               line width=0.8pt, minimum size=6mm, inner sep=0pt},
  core/.style={circle, draw=black, fill=white, line width=1.1pt,
               minimum size=11mm, inner sep=1pt, text=black, font=\small},
  blob/.style={ellipse, draw=black, fill=white, line width=1.3pt,
               inner xsep=8pt, inner ysep=7pt, text=black, font=\small},
  sep/.style={circle, draw=black, fill=white, line width=1pt,
              minimum size=8mm, inner sep=1pt, font=\small},
  pinch/.style={core},
  dia/.style={leaf},
  mrg/.style={circle, draw=black, fill=black, minimum size=3pt, inner sep=0pt},
  key/.style={circle, draw=black, fill=black, minimum size=6pt, inner sep=0pt,
              line width=1.1pt},
  dclust/.style={circle, draw=black, fill=white, line width=0.8pt,
                 minimum size=5.5mm, inner sep=0pt, font=\footnotesize},
  cclust/.style={ellipse, draw=black!55, fill=white, line width=0.8pt,
                 inner xsep=2pt, inner ysep=1.5pt, font=\scriptsize, text=black!55},
  stub/.style={draw=black, line width=0.7pt},
  cstub/.style={draw=black!55, line width=0.7pt},
  brc/.style={decorate, decoration={brace, amplitude=5pt}, line width=0.7pt},
  snapshot/.style={scale=0.95, every node/.append style={transform shape}},
}

\usepackage{etoolbox}
\usepackage{pgfplots}
\pgfplotsset{compat=1.18}

\makeatletter
\newcommand{\titleaffiliation}[1]{\gdef\@titleaffiliation{#1}}
\let\@titleaffiliation\@empty

\apptocmd{\@setauthors}{\ifx\@titleaffiliation\@empty\else
    \vspace{-0.75em}{\centering\normalfont\footnotesize
      \@titleaffiliation\par
    }\fi
}{}{}
\makeatother

\begin{document}

\title{Join Order Enumeration and Optimization For Double-Star and Helix Join Graphs}
\author{Ansh Taneja}
\author{Ryan Marcus}
\titleaffiliation{University of Pennsylvania}
\email{rcmarcus@cis.upenn.edu}

\begin{abstract}
Join ordering is classic problem in data management. We highlight the connection between structure-specific dynamic programming algorithms and the mathematical enumeration of the possible merge trees from a join graph. We present two case studies, on double-star and helix shaped join graphs, that we believe carry useful techniques for a toolbox of building structure-specific join optimizers. For the double-star case, we show how the enumeration leads directly to a significantly improved search algorithm, whereas for the helix case we show that a algorithm derived from the mathematical enumeration is not necessarily better than structure-agnostic algorithms.
\end{abstract}

\maketitle

\section{Introduction}

Join order is a fundamental problem for data management~\cite{systemr}. Good join orders lead to fast query plans, while poor ones can be disastrous~\cite{howgood}. Part of the difficulty of finding good join ordering is the diversity of possible join graphs. A common algorithm for finding optimal join orders for arbitrary join graphs takes $O(3^n)$ time~\cite{dphyp}, but if the join graph is known ahead of time, we can often do much better~\cite{path_joe}: for example, the optimal join order for a star-schema join graph can be found in $O(n^2)$ time~\cite{star_schema_jo, tree_jo_theory}.\footnote{Or in $O(n \log n)$ if we insist that the central table is the root~\cite{path_joe}.}

Designing a join ordering algorithm for a specific type of join graph can be a difficult and manual process, because different join graphs can have significantly different properties. The primary contribution of this work is two worked case studies, which we believe illustrate key techniques that can be applied to other join graphs, and highlight the connection between the mathematical enumeration of possible join orders and the dynamic programming algorithms that search that space~\cite{joe_complexity}.

Our first case study looks at ``double-star'' graphs, where we find an algorithm that significantly outperforms the classical, structure-agnostic algorithms. Our second case study looks at ``helix'' graphs, where we derive an algorithm based on the enumeration that loses out in practice to structure-agnostic algorithms.

\section{Preliminaries}\label{sec:prelim}

Before we look at any particular family, we first define the objects under study (join graphs) and our objective (enumeration).
\begin{defn}[Join graph]\label{def:joingraph}
A \emph{join graph} is a finite, simple, connected, undirected graph
$G=(V,E)$. The vertices $V$ are the relations to be joined, and $uv\in E$
whenever the query contains a join predicate relating $u$ and $v$. We write
$n=|V|$.
\end{defn}

We only consider connected join graphs (if $G$ is disconnected, then every join
order has to use a cross product somewhere) because both of the
families we study are connected by construction.

For a subset $U$ of $V$, we write $G[U]$ for the induced subgraph, and we say
that $U$ is \emph{connected} when $G[U]$ is connected. We say that two disjoint sets $U$
and $W$ are \emph{adjacent} when some edge of $G$ has one endpoint in $U$ and
the other in $W$.

\begin{defn}[Merge tree]\label{def:mergetree}
A \emph{merge tree} for a join graph $G=(V,E)$ is a rooted binary tree $T$
whose leaves are in bijection with $V$ and in which every node has either two
children or none, subject to the condition that for every node $x$ of $T$, the
set
\begin{align*}
C(x)=\{\,v\in V:\text{the leaf labelled } v \text{ lies in the subtree rooted at } x\,\}
\end{align*}
is connected in $G$. The children of a node are unordered. We call $C(x)$ the
\emph{cluster} at $x$, and we write $\mcT(G)$ for the set of merge trees for
$G$ and $f(G)=|\mcT(G)|$.
\end{defn}

A merge tree is the same thing as a join order. Each leaf is a base relation,
each internal node $x$ with children $y$ and $z$ is the join of the subplans
over $C(y)$ and $C(z)$, and the root cluster is all of $V$. We use the two
terms interchangeably. There are three important choices built into this definition:

\begin{rem}\label{rem:model}
\emph{(i) We allow bushy plans.} We put no restriction at all on the shape of
$T$; left-deep and right-deep plans are the special case where every internal
node has a leaf child.

\emph{(ii) We do not allow cross products.} This is what the connectivity
requirement on $C(x)$ does for us. We can also state that requirement locally:
at every internal node, the two child clusters must be adjacent. Indeed, if
$C(y)$ and $C(z)$ are connected and adjacent then $C(y)\cup C(z)$ is
connected, and going the other way, if $C(y)\cup C(z)$ is connected and
$C(y)$ and $C(z)$ partition it, then some edge has to cross between them. Note
that if we dropped the requirement entirely we would be left with $(2n-3)!!$
trees, a number depending only on $n$; for the $4$-cycle this gives $15$,
whereas Theorem~\ref{thm:h1} gives $10$.

\emph{(iii) We identify plans that differ by commutativity, but not by
associativity.} Since the children of a node are unordered, $R\bowtie S$ and
$S\bowtie R$ are the same merge tree. On the other hand $(R\bowtie S)\bowtie T$
and $R\bowtie(S\bowtie T)$ are different ones. The ``sided-ness'' of a join can generally be determined by local size estimates in practice (e.g., making the larger subplan the ``probe'' side of a hash join).
\end{rem}

It is often easier to read a merge tree from the bottom up, as a process: we
start with the $n$ singleton clusters, and we repeatedly replace two adjacent
clusters by their union until only one cluster is left. The next definition
gives us the language for this.

\begin{defn}[Absorption, depth]\label{def:absorb}
Let $T$ be a merge tree and let $x$ be an internal node with children $y$ and
$z$. We say that the merge at $x$ \emph{joins} $C(y)$ and $C(z)$, and for a
vertex $v\in C(y)$ we say that the cluster containing $v$ \emph{absorbs}
$C(z)$ at $x$. The \emph{depth} of a leaf $v$, written $\dep_T(v)$, is the
number of edges on the path from the root to that leaf.
\end{defn}

\begin{lemma}\label{lem:depth}
For every merge tree $T$ and every $v\in V$, the cluster containing $v$
absorbs exactly $\dep_T(v)$ clusters over the course of the process.
\end{lemma}

\begin{proof}
Look at the path from the root down to the leaf $v$. At each internal node on
this path, one child is the cluster containing $v$ and the other child is the
cluster that gets absorbed, so the clusters absorbed by $v$'s cluster
correspond exactly to the internal nodes on this path. There are $\dep_T(v)$
of them.
\end{proof}

\begin{rem}\label{rem:order}
A merge tree only fixes the relative order of two merges when one of them lies
below the other on a root-to-leaf path; merges that are incomparable are never
scheduled against each other. So whenever we speak below of the \emph{order}
in which clusters are absorbed, or of \emph{interleaving} the merges of two
subgraphs, we mean the relative order along such a path, and not a global
schedule of merges.
\end{rem}

Finally, we record the recursion that we get by looking at the root of a merge
tree. We use it as a proof device, and it is also what lets us compute counts
by brute force in Section~\ref{sec:bench}.

\begin{lemma}\label{lem:split}
For a connected join graph $G=(V,E)$ with $|V|\ge 2$,
\begin{align*}
f(G)=\frac{1}{2}\sum_{\substack{\emptyset\ne U\subsetneq V\\ U,\,V\setminus U \text{ connected}}}
f\bigl(G[U]\bigr)\,f\bigl(G[V\setminus U]\bigr),
\end{align*}
with $f(G)=1$ when $|V|=1$. The factor $\tfrac12$ is there because each root
split gets counted once as $(U,V\setminus U)$ and once as $(V\setminus U,U)$.
\end{lemma}

\section{The Graph Families}\label{sec:families}

The first graph family we will look at is the double-star graph family.

\begin{defn}\label{def:dstar}
A double star graph consists of two star graphs with distinct centers that
share exactly one spoke. An example of a double-star graph is shown in Figure~\ref{fig:dstar}.
\end{defn}

Double-star graphs can be viewed as a generalization of the common ``star schema'' join graph.

\begin{figure}[htbp]
\centering
\begin{tikzpicture}
\node[circle,draw] (a) at (3,0) {$H_1$};
\node[circle, draw] (b) at (1.5,0) {};
\node[circle, draw] (c) at (1.5,-1) {};
\node[circle, draw] (d) at (1.5,1) {};
\node[circle,draw] (e) at (5,0) {$S$};
\node[circle,draw] (f) at (7,0) {$H_2$};
\node[circle, draw] (g) at (8.5,1.5) {};
\node[circle, draw] (h) at (8.5,0.5) {};
\node[circle, draw] (i) at (8.5,-1.5) {};
\node[circle, draw] (j) at (8.5,-0.5) {};
\draw (d)--(a)--(b)--(a)--(c)--(a)--(e)--(f)--(g)--(f)--(h)--(f)--(i)--(f)--(j);
\end{tikzpicture}
\caption{An example of a double-star graph; $H_1$ and $H_2$ are the centers of
our two star graphs, and $S$ is the shared spoke.}
\label{fig:dstar}
\end{figure}
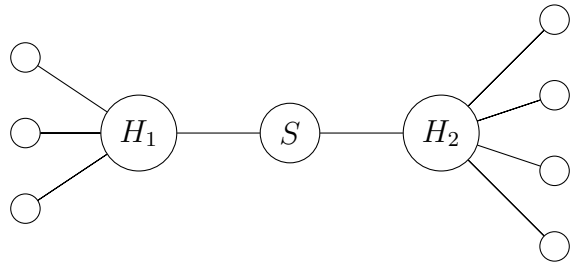

We parameterize double-star graphs by looking at the number of exclusive spokes
that each individual star has: for example, in Figure~\ref{fig:dstar}, star 1
has 3 exclusive spokes while star 2 has 4 exclusive spokes. We use the notation
$S_{a,b}$ to refer to a double-star where stars 1 and 2 have $a$ and $b$
exclusive spokes, respectively; the graph above would be $S_{3,4}$.

The second family of join graphs we consider is the helix graph.

\begin{defn}\label{def:helix}
A helix graph consists of a chain of diamonds with common pinchpoints. An example is shown in Figure~\ref{fig:helix}.

\end{defn}

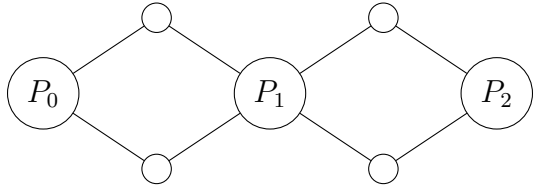
\begin{figure}[htbp]
\centering
\begin{tikzpicture}
    \node[circle,draw] (a) at (0,0) {$P_0$};
    \node[circle,draw] (b) at (1.5,1) {};
    \node[circle,draw] (c) at (1.5,-1) {};
    \node[circle,draw] (d) at (3,0) {$P_1$};
    \node[circle,draw] (e) at (4.5,1) {};
    \node[circle,draw] (f) at (4.5,-1) {};
    \node[circle,draw] (g) at (6,0) {$P_2$};
    \draw (a)--(b)--(d)--(e)--(g)--(f)--(d)--(c)--(a);
\end{tikzpicture}
\caption{An example of a helix graph; the $P_i$ are referred to as pinchpoints.
The graph shown is $\mcH_2$.}
\label{fig:helix}
\end{figure}

We refer to an $n$-helix as a helix with $n+1$ pinchpoints, or $n$ diamonds,
and we write $\mcH_n$ to refer to the $n$-helix.

\section{Join Order Enumerations}\label{sec:enum}

In this section, we demonstrate how to count the number of possible join orders for double-star and helix graphs. This serves two purposes: first, we believe our techniques, while demonstrated on these two specific classes of graphs, provide a general framework applicable to more graph structures, and second, a precise enumeration allows us to bound the optimization search space, leading to worst-case time complexity bounds for join order optimization. Our work follows a long tradition of studying the releaitonship between join graphs and their optimal enumerations~\cite{joe_complexity}.

For the double-star graph (Section~\ref{sec:enum-dstar}), we observe that the construction of a valid merge tree must follow a fixed number of ``phases,'' each of which in amenable to analysis. For the helix graph (Section~\ref{sec:enum-helix}, we demonstrate a more complex derivation by induction.

\subsection{The double star}\label{sec:enum-dstar}

We start with the double-star graph. Observe that we have four phases in any
join order of a graph of this form, as illustrated in
Figure~\ref{fig:dstar-phases}.

\begin{figure}[htbp]
\centering
\begin{subfigure}{\linewidth}\centering
\begin{tikzpicture}
  \node[core] (H1) at (0,0)   {$H_1$};
  \node[sep]  (S)  at (2.6,0) {$S$};
  \node[core] (H2) at (5.2,0) {$H_2$};
  \foreach \i/\y in {1/1,2/0,3/-1}              \node[leaf] (L\i) at (-2,\y) {};
  \foreach \i/\y in {1/1.5,2/0.5,3/-0.5,4/-1.5} \node[leaf] (R\i) at (7.2,\y) {};
  \begin{scope}[edge]
    \foreach \i in {1,2,3}   \draw (H1) -- (L\i);
    \foreach \i in {1,...,4} \draw (H2) -- (R\i);
    \draw (H1) -- (S) -- (H2);
  \end{scope}
\end{tikzpicture}
\caption{The original graph; no merges have taken place.}
\end{subfigure}

\medskip
\begin{subfigure}{\linewidth}\centering
\begin{tikzpicture}
  \node[blob] (A) at (0,0)   {$H_1,\,S,\,\dots$};
  \node[blob] (B) at (5.2,0) {$H_2,\,\dots$};
  \node[leaf] (L)  at (-3,0)     {};
  \node[leaf] (R1) at (7.4,0.6)  {};
  \node[leaf] (R2) at (7.4,-0.6) {};
  \begin{scope}[edge]
    \draw (L) -- (A) -- (B);
    \draw (B) -- (R1);  \draw (B) -- (R2);
  \end{scope}
\end{tikzpicture}
\caption{After step 1: $S$ has been absorbed into $H_1$.}
\end{subfigure}

\medskip
\begin{subfigure}{\linewidth}\centering
\begin{tikzpicture}
  \node[blob] (A) at (0,0) {$H_1,\,S,\,H_2,\,\dots$};
  \node[leaf] (L)  at (-3.4,0)   {};
  \node[leaf] (R1) at (3.4,0.7)  {};
  \node[leaf] (R2) at (3.4,-0.7) {};
  \begin{scope}[edge]
    \draw (L) -- (A);  \draw (A) -- (R1);  \draw (A) -- (R2);
  \end{scope}
\end{tikzpicture}
\caption{After step 2: $H_2$ has been absorbed as well.}
\end{subfigure}

\medskip
\begin{subfigure}{\linewidth}\centering
\begin{tikzpicture}
  \node[blob] (A) at (0,0) {Fully absorbed};
\end{tikzpicture}
\caption{After step 3: everything has been absorbed.}
\end{subfigure}
\caption{Successive merges absorb the whole graph into a single vertex.}
\label{fig:dstar-phases}
\end{figure}
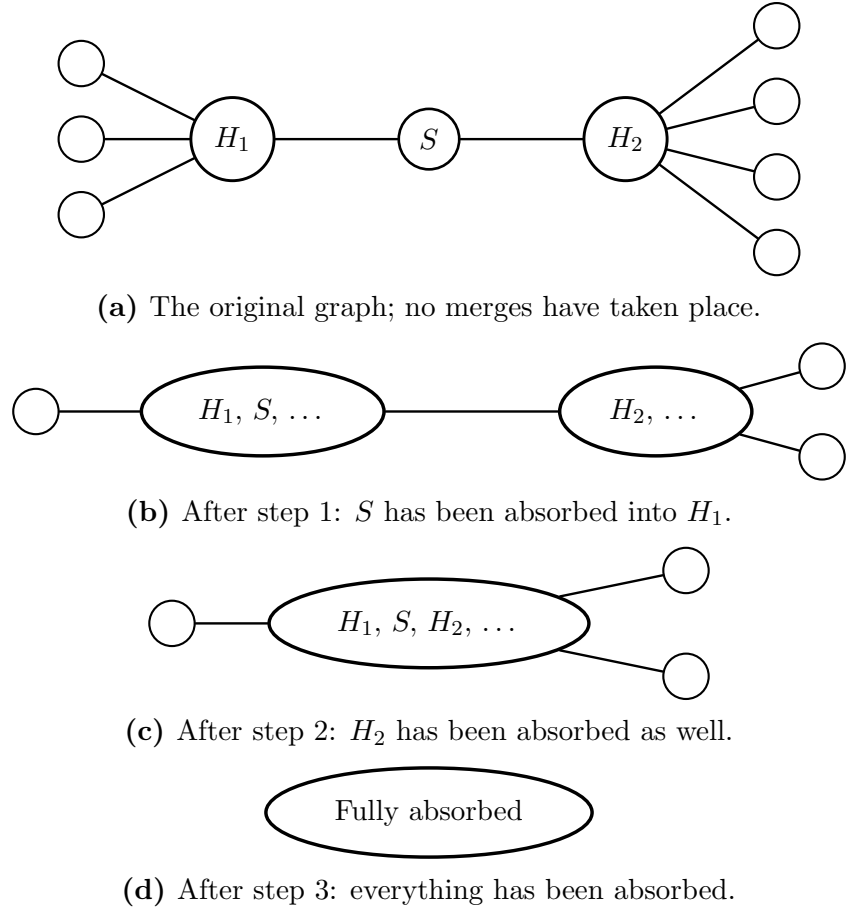

\begin{prop}\label{prop:dstar-sum}
The number of distinct join orders for the graph $S_{a,b}$ is
\begin{align*}
f(S_{a,b})=\sum_{a_1\le a,\text{ }b_1 \le b}
\left[\binom{a}{a_1}a_1!\binom{b}{b_1}b_1!(a+b-a_1-b_1)!\right](a_1+b_1+2).
\end{align*}
\end{prop}

\begin{proof}
Note that we could also have $H_2$ merge with $S$ first, and then merge with
$H_1$.

Noting this, we can condition on the number of nodes merged before step 2 as
follows.

Fix $S_{a,b}$, and suppose $H_1$ merges with $S$ first. Now before step 2
happens, say that $a_1+1$ nodes (including $S$) get absorbed into $H_1$, and
$b_1$ nodes get absorbed into $H_2$. There are
\begin{align*}
    \binom{a}{a_1}(a_1+1)!\binom{b}{b_1}b_1!
\end{align*}
ways to choose these nodes and join them.

Now, once step 2 happens (there is only one way to do it, so it doesn't affect
our total), we have a star graph with $a-a_1+b-b_1$ nodes; there are
\begin{align*}
    (a+b-a_1-b_1)!
\end{align*}
possible join orders from here. In total, we have
\begin{align*}
    \sum_{a_1\le a,\text{ }b_1 \le b} \binom{a}{a_1}(a_1+1)!\binom{b}{b_1}b_1!(a+b-a_1-b_1)!
\end{align*}
join orders for this case.

Now, we still have one more case: namely, we have not considered the
possibility that $H_2$ was joined with $S$ before $H_1$. However, this case is
nearly identical to the previous one, as we can just swap $(a,a_1)$ and
$(b,b_1)$. Doing so, there are
\begin{align*}
    \sum_{a_1\le a,\text{ }b_1 \le b} \binom{a}{a_1}a_1!\binom{b}{b_1}(b_1+1)!(a+b-a_1-b_1)!
\end{align*}
join orders for this case, and in total, there are
\begin{align*}
    f(S_{a,b})&=\sum_{a_1\le a,\text{ }b_1 \le b} \left[\binom{a}{a_1}a_1!\binom{b}{b_1}b_1!(a+b-a_1-b_1)!\right](a_1+1+b_1+1)\\
    &=\sum_{a_1\le a,\text{ }b_1 \le b} \left[\binom{a}{a_1}a_1!\binom{b}{b_1}b_1!(a+b-a_1-b_1)!\right](a_1+b_1+2)
\end{align*}
join orders for a $S_{a,b}$ graph.
\end{proof}

\begin{rem}\label{rem:dstar-growth}
For our purposes, this is enough --- the number of possible join orders clearly
grows faster than any simple exponential function. However, through additional
algebraic manipulations, we can arrive at a simple closed form.
\end{rem}

\begin{thm}\label{thm:dstar}
The number of distinct join orders for the graph $S_{a,b}$ is
\begin{align*}
    f(S_{a,b})=a!b!\left[\binom{a+b+4}{b+2}-a-b-4\right].
\end{align*}
\end{thm}

\begin{proof}
We start from Proposition~\ref{prop:dstar-sum} and first break each binomial
coefficient into factorials. This lets us simplify our summation to
\begin{align*}
    f(S_{a,b})&=\sum_{a_1\le a,\text{ }b_1 \le b}\frac{a!b!(a+b-a_1-b_1)!}{(a-a_1)!(b-b_1)!}(a_1+b_1+2)\\
    &= a!b!\sum_{a_1\le a,\text{ }b_1 \le b}\binom{a+b-a_1-b_1}{a-a_1}(a_1+b_1+2).
\end{align*}
Now, let $x = a-a_1$ and $y=b-b_1$. Our sum now turns into
\begin{align*}
    f(S_{a,b})&= a!b!\sum_{x\le a,\text{ }y \le b}\binom{x+y}{x}(a+b-x-y+2),
\end{align*}
which we split as follows:
\begin{align*}
    f(S_{a,b})&= a!b!\left[(a+b+2)\sum_{x\le a,\text{ }y \le b}\binom{x+y}{x}-\sum_{x\le a,\text{ }y \le b}\binom{x+y}{x}(x+y)\right].
\end{align*}
Now, focus on the first summation. We can initially simplify it by splitting it
into two summations and using the Hockey Stick Identity:
\begin{align*}
    &\sum_{x=0}^a \sum_{y=0}^b \binom{x+y}{x}\\
    =&\sum_{x=0}^a\binom{x+b+1}{x+1}.
\end{align*}
Now, we direct our attention to the second summation,
\begin{align*}
    \sum_{x=0}^a\sum_{y=0}^b\binom{x+y}{x}(x+y).
\end{align*}
We expand the inner part of this summation as
\begin{align*}
   &x\binom{x}{x}+(x+1)\binom{x+1}{x}+(x+2)\binom{x+2}{x}+\cdots+(x+b)\binom{x+b}{x}
\end{align*}
\begin{align*}
    =(x+1)\binom{x}{x}+(x+2)\binom{x+1}{x}+(x+3)\binom{x+2}{x}+\cdots+(x+b+1)\binom{x+b}{x}
\end{align*}
\begin{align*}
    -\binom{x}{x}-\binom{x+1}{x}-\binom{x+2}{x}-\cdots-\binom{x+b}{x}.
\end{align*}
Now, every positive term on the resulting expression is of the form
\begin{align*}
    (x+n+1)\binom{x+n}{x} = \frac{(x+n+1)!}{x!n!}=(x+1)\binom{x+n+1}{x+1}.
\end{align*}
This lets us rewrite our expression as
\begin{align*}
    (x+1)\binom{x+1}{x+1}+(x+1)\binom{x+2}{x+1}+(x+1)\binom{x+3}{x+1}+\cdots+(x+1)\binom{x+b+1}{x+1}
\end{align*}
\begin{align*}
    -\binom{x}{x}-\binom{x+1}{x}-\binom{x+2}{x}-\cdots-\binom{x+b}{x}\\
\end{align*}
\begin{align*}
    =(x+1)\binom{x+b+3}{x+2}-\binom{x}{x}-\binom{x+1}{x}-\binom{x+2}{x}-\cdots-\binom{x+b}{x}\\
\end{align*}
\begin{align*}
    =(x+1)\binom{x+b+2}{x+2}-\binom{x+b+1}{x+1}.
\end{align*}
Now, we go back to the original expression,
\begin{align*}
    f(S_{a,b}) &= a!b!\left[(a+b+2)\sum_{x=0}^a\binom{x+b+1}{x+1}-\sum_{x=0}^a(x+1)\binom{x+b+2}{x+2}+\sum_{x=0}^a\binom{x+b+1}{x+1}\right]\\
    &= a!b!\left[(a+b+3)\sum_{x=0}^a\binom{x+b+1}{x+1}-\sum_{x=0}^a(x+1)\binom{x+b+2}{x+2}\right]\\
    &= a!b!\left[(a+b+3)\sum_{x=0}^a\binom{x+b+1}{b}-\sum_{x=0}^a(x+1)\binom{x+b+2}{x+2}\right].
\end{align*}
Now, look at the summation
\begin{align*}
    \sum_{x=0}^a(x+1)\binom{x+b+2}{x+2}.
\end{align*}
We can express this as
\begin{align*}
    &\sum_{x=0}^a(x+2)\binom{x+b+2}{x+2}-\sum_{x=0}^a \binom{x+b+2}{x+2}\\
    = &\sum_{x=0}^a \frac{(x+b+2)!}{(x+1)!b!}-\sum_{x=0}^a \binom{x+b+2}{b}\\
    = &\sum_{x=0}^a (b+1)\binom{x+b+2}{b+1}-\left[\binom{a+b+3}{b+1}-\binom{b}{b}-\binom{b+1}{b}\right]\\
    = &{}(b+1)\binom{a+b+3}{b+2}-\binom{a+b+3}{b+1}+1.
\end{align*}
This means that
\begin{align*}
    f(S_{a,b}) &= a!b!\left[(a+b+3)\sum_{x=0}^a\binom{x+b+1}{b}-\left[(b+1)\binom{a+b+3}{b+2}-\binom{a+b+3}{b+1}+1\right]\right]\\
    &=  a!b!\left[(a+b+3)\left[\binom{a+b+2}{b+1}-1\right]-\left[(b+1)\binom{a+b+3}{b+2}-\binom{a+b+3}{b+1}+1\right]\right]\\
    &= a!b!\left[(a+b+3)\binom{a+b+2}{b+1}-(a+b+3)-(b+1)\binom{a+b+3}{b+2}+\binom{a+b+3}{b+1}-1\right].
\end{align*}
Now, we have that
\begin{align*}
    (a+b+3)\binom{a+b+2}{b+1}=(b+2)\binom{a+b+3}{b+2},
\end{align*}
so
\begin{align*}
    f(S_{a,b})=a!b!\left[\binom{a+b+3}{b+2}+\binom{a+b+3}{b+1}-a-b-4\right].
\end{align*}
Now, using Pascal's Identity, we arrive at the stated result.
\end{proof}

\begin{rem}\label{rem:dstar-sym}
The double star $S_{a,b}$ is isomorphic to $S_{b,a}$, so $f$ must be symmetric
in $a$ and $b$; the closed form of Theorem~\ref{thm:dstar} is, since
$\binom{a+b+4}{b+2}=\binom{a+b+4}{a+2}$.
\end{rem}

\subsection{The helix}\label{sec:enum-helix}

Now, we move on to the helix.

The helix is different from the double star in that we cannot focus on any one
absorption moment. Instead, we proceed by induction on $n$. We start by looking
at $\mcH_1$, the case of a single diamond, whose vertices are the two
pinchpoints $P_0$ and $P_1$ together with $A_1$ and $B_1$. Since $A_1$ and
$B_1$ are not adjacent, they can never merge with each other directly
(Remark~\ref{rem:model}(ii)).

We can look at the absorption process from $P_0$'s point of view. Namely, we
look at how many independent clusters the cluster containing $P_0$ absorbs
before the graph is fully decomposed. For example, if $P_0$ absorbs everything
in one cluster, that means that the three other vertices must have merged into
a single cluster before any of them interacted with $P_0$.

\begin{thm}\label{thm:h1}
The number of valid join orders for $\mcH_1$ is $10$; that is, $f(\mcH_1)=10$.
\end{thm}

\begin{proof}
We case on the number of clusters that $P_0$ absorbs.

Our first case is if $P_0$ absorbs one cluster. In this case $P_1$, $A_1$, and
$B_1$ all merged together first, and then that cluster merged with $P_0$. There
are two ways this could have happened: the order of absorption within that
cluster could have been
\begin{align*}
    P_1\rightarrow A_1\rightarrow B_1\rightarrow P_0
\end{align*}
or
\begin{align*}
    P_1\rightarrow B_1\rightarrow A_1\rightarrow P_0.
\end{align*}
Therefore, the subtotal for this case is $2$.

Our next case is if $P_0$ absorbs two clusters. Here, one of the two clusters
must contain two vertices and the other must consist of a single vertex. There
are two possibilities for the cluster of size two, since $A_1$ cannot merge
with $B_1$, and once that cluster is chosen the cluster of size one is fixed.
Afterwards, $P_0$ can absorb these two clusters in $2!=2$ different orders.
Therefore, the subtotal for this case is $2\cdot 2=4$.

Now, the final case is if $P_0$ absorbs three clusters. This means that it
absorbs each of the other vertices independently, so the subtotal for this case
is $3!-2!=4$: we count all the different orders of absorption, then subtract
those in which $P_1$ is the vertex absorbed first, which is impossible because
$P_1$ is not adjacent to $P_0$.

Adding the three subtotals gives $2+4+4=10$.
\end{proof}

We formalize this counting scheme.

\begin{defn}\label{def:E}
For $n\ge 1$ and $d\ge 1$, let $E_n(d)$ be the number of merge trees of
$\mcH_n$ in which the cluster containing the end pinch point $P_0$ absorbs
exactly $d$ clusters over the course of the process. Equivalently
(Lemma~\ref{lem:depth}), $E_n(d)$ counts the merge trees in which the leaf
$P_0$ sits at depth $d$.
\end{defn}

By construction $f(\mcH_n)=\sum_{d\ge 1}E_n(d)$, a finite sum since
$1\le\dep_T(P_0)\le 3n$, and the proof of Theorem~\ref{thm:h1} shows that
\begin{align*}
    E_1(1)=2,\qquad E_1(2)=4,\qquad E_1(3)=4.
\end{align*}
Using this function, we can compute the number of valid join orders for
$\mcH_2$.

\begin{thm}\label{thm:h2}
The number of valid join orders for $\mcH_2$ is $856$; that is,
$f(\mcH_2)=856$.
\end{thm}

\begin{proof}
We look at the merges from the point of view of the middle pinch point $P_1$.
It is an end pinch point for both of the diamonds in the graph, so we may apply
$E_1$ to each of them separately.

In particular, any merge not involving $P_1$ is local: a merge inside one
diamond does not affect the other diamond. Therefore we can focus on just the
merges involving the cluster containing $P_1$.

Suppose that the first diamond is absorbed into that cluster in $\alpha$
clusters and the second is absorbed in $\beta$ clusters. The relative order of
the merges contributed by a single diamond is already accounted for by $E_1$,
so all that is left to do is to interleave the $\alpha$ merges coming from the
first diamond with the $\beta$ merges coming from the second. Every one of
these merges involves the cluster containing $P_1$, so distinct interleavings
give distinct merge trees, and there are $\binom{\alpha+\beta}{\alpha}$ ways to
carry it out. The total is therefore
\begin{align*}
    f(\mcH_2)=\sum_{1\leq \alpha,\beta \leq 3} E_1(\alpha)\,E_1(\beta)\binom{\alpha+\beta}{\alpha},
\end{align*}
and summing over all nine pairs gives $856$.
\end{proof}

This approach does not work for larger $n$, since there is no single pinch
point that is an end pinch point of every diamond. However, we can build
$\mcH_n$ out of $\mcH_{n-1}$ one diamond at a time, and the key observation is
that the cost of gluing on one more diamond does not depend on how long the
existing chain is. This motivates the following definition.

\begin{defn}\label{def:K}
Let $K(a,d)$ denote the number of ways of extending a merge tree of
$\mcH_{n-1}$ counted by $E_{n-1}(a)$ to a merge tree of $\mcH_n$ counted by
$E_n(d)$; that is, the number of ways of attaching one new diamond and
interleaving its merges so that the end pinch point of the longer helix absorbs
$d$ clusters in total.
\end{defn}

Because gluing on a diamond is a purely local operation, $K(a,d)$ is
independent of $n$, and so
\begin{align}\label{eq:recursion}
    E_n(d) = \sum_{a} E_{n-1}(a)\,K(a,d)
\end{align}
holds for all $n\ge 2$. Adding one more diamond can move the absorption count
in either direction. The new end pinch point absorbs somewhere between one and
three clusters made out of the diamond's own vertices, together with however
many of the $a$ existing chain clusters end up above the gluing point, and that
last number can be anything from $0$ to $a$. So $d$ ranges from $1$ to $a+3$,
and $K(a,d)=0$ outside that range; for a fixed $d$, the terms contributing to
\eqref{eq:recursion} are those with $a\ge d-3$, with no upper bound on $a$. It
remains to determine $K(a,d)$.

\begin{thm}\label{thm:K}
For integers $a\ge 0$ and $d\ge 1$,
\[
K(a,d)=
\begin{cases}
(a+1)(a+2), & d=1,\\[2pt]
(a+1)(a+4), & d=2,\\[2pt]
(a+1)(a+4)-d(d-3), & 3\le d\le a+3,\\[2pt]
0, & \text{otherwise.}
\end{cases}
\]
\end{thm}

\begin{proof}
We obtain an expression for $K(a,d)$ by viewing the situation as one new
diamond glued onto a longer, already-determined chain. Write $P$ for the pinch
point at which the new diamond is glued to the existing chain, write $A$ and
$B$ for the two middle vertices of the new diamond, and write $Z$ for its new
far endpoint, which is the end pinch point of the longer helix. The situation
is shown in Figure~\ref{fig:glue}.

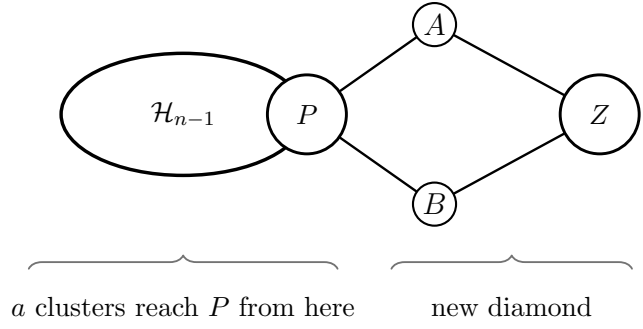
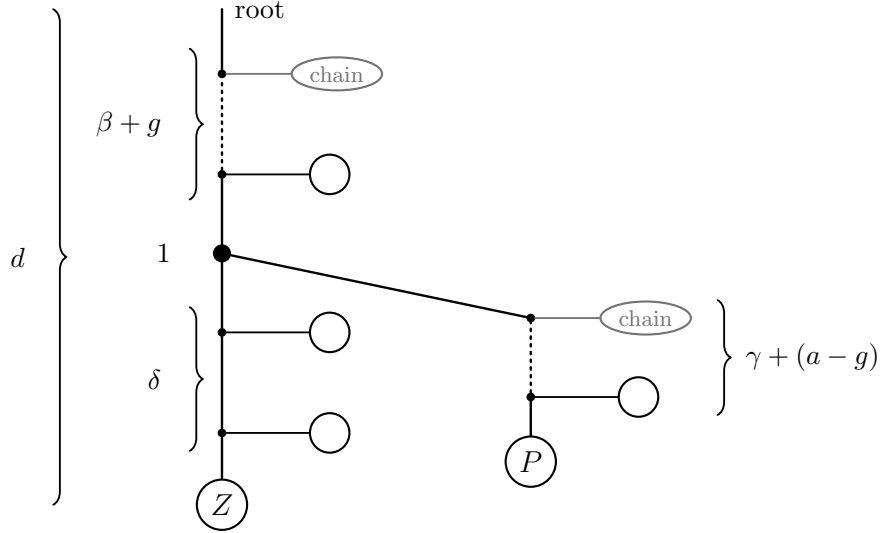
\begin{figure}[htbp]
\centering
\begin{subfigure}{\linewidth}\centering
\begin{tikzpicture}[snapshot]
  \node[blob, minimum width=34mm, minimum height=17mm] (chain) at (0,0) {$\mcH_{n-1}$};
  \node[pinch] (P) at (chain.east) {$P$};
  \node[dia]   (A) at (3.5,1.25) {$A$};
  \node[dia]   (B) at (3.5,-1.25) {$B$};
  \node[pinch] (Z) at (5.8,0) {$Z$};
  \begin{scope}[edge]
    \draw (P) -- (A) -- (Z) -- (B) -- (P);
  \end{scope}
  \draw[brc, black!55] (2.9,-2.15) -- (6.25,-2.15)
        node[midway, below=7pt, font=\small, black]{new diamond};
  \draw[brc, black!55] (-2.15,-2.15) -- (2.15,-2.15)
        node[midway, below=7pt, font=\small, black]{$a$ clusters reach $P$ from here};
\end{tikzpicture}
\caption{One new diamond glued onto the chain at its end pinch point $P$.}
\label{fig:glue-a}
\end{subfigure}

\vspace{1.2em}

\begin{subfigure}{\linewidth}\centering
\begin{tikzpicture}[snapshot]
  \node[leaf, minimum size=7mm] (Zl) at (0,0) {$Z$};
  \node[mrg] (m1) at (0,1.0) {};
  \node[mrg] (m2) at (0,2.4) {};
  \node[key] (mp) at (0,3.5) {};
  \node[mrg] (m3) at (0,4.6) {};
  \node[mrg] (m4) at (0,6.0) {};

  \begin{scope}[edge]
    \draw (Zl) -- (m1) -- (m2) -- (mp) -- (m3);
    \draw (m4) -- (0,6.9);
  \end{scope}
  \draw[edge, dotted] (m3) -- (m4);
  \node[font=\small, right=1pt] at (0,6.9) {root};

  \node[dclust] (s1) at (1.5,1.0) {};
  \node[dclust] (s2) at (1.5,2.4) {};
  \node[dclust] (s3) at (1.5,4.6) {};
  \node[cclust] (s4) at (1.6,6.0) {chain};
  \begin{scope}[stub]
    \draw (m1) -- (s1); \draw (m2) -- (s2); \draw (m3) -- (s3);
  \end{scope}
  \draw[cstub] (m4) -- (s4);

  \draw[brc] (-0.45,2.75) -- (-0.45,0.75)
        node[midway, left=7pt, font=\small]{$\delta$};
  \draw[brc] (-0.45,6.35) -- (-0.45,4.25)
        node[midway, left=7pt, font=\small]{$\beta+g$};
  \node[left=2pt, font=\small] at (-0.5,3.5) {$1$};

  \draw[brc] (-2.35,6.9) -- (-2.35,0)
        node[midway, left=7pt, font=\small]{$d$};

  \node[leaf, minimum size=7mm] (Pl) at (4.3,0.6) {$P$};
  \node[mrg] (n1) at (4.3,1.5) {};
  \node[mrg] (n2) at (4.3,2.6) {};
  \begin{scope}[edge]
    \draw (Pl) -- (n1);
    \draw (n2) -- (mp);
  \end{scope}
  \draw[edge, dotted] (n1) -- (n2);
  \node[dclust] (t1) at (5.8,1.5) {};
  \node[cclust] (t2) at (5.9,2.6) {chain};
  \draw[stub] (n1) -- (t1);
  \draw[cstub] (n2) -- (t2);

  \draw[brc] (6.9,2.85) -- (6.9,1.25)
        node[midway, right=7pt, font=\small]{$\gamma+(a-g)$};
\end{tikzpicture}
\caption{The path from the leaf $Z$ up to the root of the merge tree.}
\label{fig:glue-b}
\end{subfigure}

\caption{Gluing one diamond onto $\mcH_{n-1}$, and the effect on the depth of
$Z$. In (b) the bold node is the merge at which $Z$'s cluster and $P$'s cluster
come together; the merges below it on the spine are the $\delta$ merges in
which $Z$ absorbs $A$ and\slash or $B$, and the merges above it are the $\beta$
remaining diamond clusters interleaved with the $g$ chain clusters placed above
$P$. The branch on the right collects the $\gamma$ diamond clusters and the
$a-g$ chain clusters that $P$'s cluster absorbs before meeting $Z$; these lie
off the root-to-$Z$ path and so leave $d$ unchanged. Circles denote clusters of
diamond vertices and grey ellipses denote clusters coming from the rest of the
chain. The figure shows the shape of the tree rather than one particular
pattern.}
\label{fig:glue}
\end{figure}

As the process runs, the cluster containing $Z$ and the cluster containing $P$
both grow, and at some point the two merge with each other. The clusters
absorbed along the way are the ones made up of the diamond's three new vertices
$A$, $B$, and $Z$, grouped into connected clusters, together with the $a$
clusters contributed by the rest of the chain, every one of which reaches $Z$
through $P$. A \emph{pattern} records how the merge tree behaves on the diamond
alone: how $A$, $B$, and $Z$ are grouped into connected clusters, and where
those clusters sit relative to the merge of the $Z$-cluster with the
$P$-cluster. It is summarized by a triple $(\alpha,\beta,\gamma)$, where

\begin{itemize}
\item $\alpha$ is the depth of $Z$ in the diamond-only merge tree, that is, the
number of merges on the path from $Z$ up to the root. This is also the number
of clusters that $Z$'s own cluster absorbs, and it is the value of $d$
\emph{before} any chain clusters are mixed in. Note that being absorbed
\emph{later} makes $\alpha$ \emph{smaller}: a vertex absorbed in the final
merge sits just one step below the root, whereas one absorbed early sits deep
in the tree.
\item $\beta$ is the number of the diamond's clusters that are absorbed
\emph{after} the cluster containing $P$, that is, the merges lying \emph{above}
$P$ on its path to the root. Inserting a chain cluster into one of these slots
pushes $Z$ one level deeper.
\item $\gamma$ is the number of the diamond's clusters that are absorbed
\emph{before} the cluster containing $P$, that is, the merges lying
\emph{below} $P$. Inserting a chain cluster into one of these slots leaves the
depth of $Z$ unchanged.
\end{itemize}

These quantities satisfy $\beta+\gamma+1$ equal to the number of clusters the
diamond is split into, and $\alpha=\delta+1+\beta$, where $\delta\in\{0,1,2\}$
is the depth of $Z$ inside its own cluster, which is $2$ when $Z$ lies in a
triple, $1$ when it lies in a pair, and $0$ when it is alone. Since the diamond
contributes only the three vertices $A$, $B$, $Z$, we have
$\delta+\beta+\gamma\le 2$, and so there are exactly five patterns, each of
weight $2$ by the symmetry between $A$ and $B$:
\[
\begin{array}{c|ccc|c}
\# & \alpha & \beta & \gamma & \text{weight}\\\hline
1 & 3 & 0 & 0 & 2\\
2 & 3 & 1 & 0 & 2\\
3 & 2 & 0 & 1 & 2\\
4 & 1 & 0 & 2 & 2\\
5 & 2 & 1 & 1 & 2
\end{array}
\]

The $a$ existing chain clusters are then interleaved with the new diamond's
merges. Writing $g=d-\alpha$ for the number of them placed above $P$, and $a-g$
for the number placed below, the count is
\begin{align*}
K(a,d)=\sum_{\text{patterns}}\mathrm{weight}\cdot
\binom{\beta+g}{g}\binom{\gamma+(a-g)}{a-g},
\qquad g=d-\alpha,
\end{align*}
each summand being present only when $0\le g\le a$, that is, when
$\alpha\le d\le \alpha+a$.

Substituting $g=d-\alpha$ into the equation for $K(a,d)$ obtained above and
simplifying the binomial coefficients via $\binom{n}{0}=1$, $\binom{n}{1}=n$,
and $\binom{n}{2}=\tfrac{n(n-1)}{2}$ gives the five contributions
\begin{align*}
T_1 &= 2\binom{g}{g}\binom{a-g}{a-g} = 2, & (3\le d\le a+3),\\
T_2 &= 2\binom{1+g}{1}\binom{a-g}{a-g} = 2(d-2), & (3\le d\le a+3),\\
T_3 &= 2\binom{g}{g}\binom{a-g+1}{1} = 2(a-d+3), & (2\le d\le a+2),\\
T_4 &= 2\binom{g}{g}\binom{a-g+2}{2} = (a-d+3)(a-d+2), & (1\le d\le a+1),\\
T_5 &= 2\binom{1+g}{1}\binom{a-g+1}{1} = 2(d-1)(a-d+3), & (2\le d\le a+2).
\end{align*}
From the lower endpoints of these ranges, the set of active terms is $T_4$ only
for $d=1$; $T_3,T_4,T_5$ for $d=2$; and all five for $d\ge 3$.

\medskip
\noindent\emph{Case $d\ge 3$.} Summing all five terms, we group the constant
and linear parts as follows:
\begin{align*}
T_1+T_2+T_3 &= 2 + 2(d-2) + 2(a-d+3)\\
            &= 2 + 2d - 4 + 2a - 2d + 6\\
            &= 2a+4.
\end{align*}
Then we factor the common $(a-d+3)$ out of the quadratic parts:
\begin{align*}
T_4+T_5 &= (a-d+3)(a-d+2) + 2(d-1)(a-d+3)\\
        &= (a-d+3)\bigl[(a-d+2)+2(d-1)\bigr]\\
        &= (a-d+3)(a+d).
\end{align*}
Adding the two pieces,
\begin{align*}
K(a,d) &= (2a+4) + (a-d+3)(a+d)\\
       &= (2a+4) + \bigl(a^2 - d^2 + 3a + 3d\bigr)\\
       &= a^2 + 5a + 4 - d^2 + 3d\\
       &= (a+1)(a+4) - d(d-3).
\end{align*}
Although $T_4$ at $d=a+2$, and $T_3$ and $T_5$ at $d=a+3$, leave their stated
ranges before $d$ reaches $a+3$, each carries a factor of $(a-d+3)$ or
$(a-d+2)$ that vanishes there, so including them anyway does not change the
total. Therefore the identity holds on the whole range $3\le d\le a+3$, and
$K(a,d)=0$ for $d>a+3$.

\medskip
\noindent\emph{Case $d=2$.} Only $T_3,T_4,T_5$ are active:
\begin{align*}
K(a,2) &= T_3 + T_4 + T_5\\
       &= 2(a+1) + (a+1)a + 2(a+1)\\
       &= (a+1)(2 + a + 2)\\
       &= (a+1)(a+4).
\end{align*}

\medskip
\noindent\emph{Case $d=1$.} Only $T_4$ is active:
\begin{align*}
K(a,1) = T_4\big|_{d=1} = (a+1)(a+2).
\end{align*}

Collecting the three cases, together with $K(a,d)=0$ for $d<1$ or $d>a+3$,
yields the stated piecewise formula.
\end{proof}

Note that $\mcH_1$ is a single diamond glued onto an empty chain, so
$E_1(d)=K(0,d)$; indeed $K(0,1)=2$, $K(0,2)=4$, and $K(0,3)=4$, matching
Theorem~\ref{thm:h1}. Using this, we can unravel the recursion for $E_n(d)$ and
obtain a summation for the whole count:

\begin{thm}\label{thm:hn}
The number of valid join orders for $\mcH_n$ is
\begin{align*}
f(\mcH_n)=\sum_{\substack{a_0=0\\ 1\le a_i\le a_{i-1}+3}}\ \prod_{i=1}^{n}K(a_{i-1},a_i).
\end{align*}
\end{thm}

\begin{proof}
Expanding \eqref{eq:recursion} repeatedly from $E_n$ down to $E_1$, and using
$E_1(a_1)=K(a_0,a_1)$ with $a_0=0$, writes $E_n(a_n)$ as the sum of
$\prod_{i=1}^{n}K(a_{i-1},a_i)$ over all chains $a_0=0,a_1,\dots,a_n$ with
$1\le a_i\le a_{i-1}+3$, the range coming from the support of $K$ in
Theorem~\ref{thm:K}. Summing over $a_n$ and using
$f(\mcH_n)=\sum_{d\ge 1}E_n(d)$ gives the result.
\end{proof}

\section{Optimal Join Ordering Algorithms}\label{sec:algorithms}

The counts of Section~\ref{sec:enum} tell us that we cannot simply enumerate
the plan space. However, they also tell us how to avoid having to. We were
able to count the merge trees of these families in the first place because
those trees have a very restricted shape, and the question now is how much of
that shape we can put to work. The two families answer it differently. For the
double star the shape cuts the search down to a handful of candidates and we
get a polynomial algorithm. For the helix it does not, and what we get instead
is an exact analysis of the generic dynamic program. We take the two cases in
turn, and Section~\ref{sec:bench} measures both.

\subsection{The cost model}\label{sec:cost}

So far our join graphs have carried no numbers, so we start by attaching some. Our formulation closely follows~\cite{cross_products_joe}.

\begin{defn}[Weighted join graph]\label{def:weighted}
A \emph{weighted join graph} is a join graph $G=(V,E)$ together with a
\emph{weight} $w_v>0$ for each $v\in V$ and a \emph{selectivity}
$\sigma_e>0$ for each $e\in E$. For a connected $U\subseteq V$ we write
\begin{align*}
\card(U) \;=\; \prod_{v\in U} w_v \prod_{e\in E(G[U])} \sigma_e
\end{align*}
for the estimated cardinality of the join of $U$. We keep $|U|$ for the number
of vertices in $U$; these are two different quantities and we never write one
for the other.
\end{defn}

\begin{defn}[Cost of a merge tree]\label{def:cout}
The \emph{cost} of a merge tree $T$ for a weighted join graph is
\begin{align*}
\mathrm{cost}(T)=\sum_{x \text{ an internal node of } T} \card\bigl(C(x)\bigr).
\end{align*}
A plan is \emph{optimal} if it attains $\min_{T\in\mcT(G)}\mathrm{cost}(T)$.
\end{defn}

\begin{rem}\label{rem:cout}
This is the cost function that is usually called $C_{\mathrm{out}}$: the cost
of a plan is the total size of the intermediate results that it materializes.
Leaves cost nothing, since scanning the base relations is something every plan
has to do. Note that our cardinality model assumes that the selectivities are
independent; this is what lets us write $\card(U)$ as a function of the vertex
set alone, rather than of the tree that produced it.
\end{rem}

\begin{rem}\label{rem:F}
Note that the root cluster is $V$ for every merge tree, so
\begin{align*}
F \;=\; \card(V) \;=\; \prod_{v\in V} w_v \prod_{e\in E}\sigma_e
\end{align*}
is the same for every plan, and it shows up as a summand of $\mathrm{cost}(T)$
no matter which $T$ we pick. We use this in Lemma~\ref{lem:finalsize}.
\end{rem}

For the rest of this section, $H_1$ and $H_2$ are the hubs of $S_{a,b}$,
$\mcA$ is the set of $a$ exclusive spokes of $H_1$, $\mcB$ is the set of $b$
exclusive spokes of $H_2$, and $S$ is the shared spoke.

\begin{defn}[Fanout]\label{def:fanout}
For a spoke $u$ of a hub $h$, the \emph{fanout} of $u$ is
$\varphi_u = w_u\,\sigma_{hu}$. For the shared spoke we write
$\varphi^{(1)}_S=w_S\sigma_{H_1S}$ and $\varphi^{(2)}_S=w_S\sigma_{SH_2}$.
\end{defn}

Absorbing a spoke $u$ into a cluster multiplies the cardinality of that
cluster by $\varphi_u$. So a spoke whose fanout is less than $1$ shrinks the
running result, while a spoke whose fanout is greater than $1$ grows it.
Everything our algorithm computes is built out of these fanouts.

\subsection{Absorbing spokes one at a time}\label{sec:rank}

In a double star, both hubs spend the whole plan doing the same thing, namely
absorbing single spokes one after another. We now record the two facts we need
about this operation. The first of them is the classical rank-ordering rule.

\begin{lemma}[Rank ordering]\label{lem:rank}
Let a cluster of cardinality $N$ absorb singletons $u_1,\dots,u_m$ one at a
time, in the order given by a permutation $\pi$. The cost contributed is
\begin{align*}
N\sum_{k=1}^{m}\ \prod_{i\le k}\varphi_{\pi(i)},
\end{align*}
and it is minimized by taking $\pi$ to list the $u_i$ in ascending order of
fanout. The resulting cluster has cardinality $N\prod_{i}\varphi_{u_i}$
regardless of $\pi$.
\end{lemma}

\begin{proof}
Absorbing the first $k$ singletons produces a cluster of cardinality
$N\prod_{i\le k}\varphi_{\pi(i)}$, and the cost is the sum of these over
$k=1,\dots,m$. The final cardinality is the $k=m$ term, which does not depend
on $\pi$.

For the optimality claim we proceed by induction on $m$, proving the statement
for every cluster cardinality $N$ at once.

The case $m=1$ is trivial; there is only one order, and it is the desired one.

Now, say the claim holds for $m$ singletons, and consider a cluster that has
$m+1$ of them to absorb. Label the singletons so that
$\varphi_1\le\varphi_2\le\cdots\le\varphi_{m+1}$. Any order is determined by
which singleton gets absorbed first together with an order of the remaining
$m$. If we absorb $u_j$ first, we are left with a cluster of cardinality
$N\varphi_j$ and the other $m$ singletons untouched, so we may apply the
inductive hypothesis with $N\varphi_j$ in place of $N$; it tells us that the
cheapest way to finish is to take the rest in ascending order. So it is enough
to compare the $m+1$ candidates we get by varying the first singleton.

Then, say we absorb $u_1$ first. The total cost for this order is
\begin{align*}
a &= N\varphi_1\Bigl(1+\varphi_2+\varphi_2\varphi_3+\cdots+\prod_{i=2}^{m+1}\varphi_i\Bigr).
\end{align*}
Now, say we absorb $u_\ell$ first, where $\ell\ne 1$; then the cost of that
order is
\begin{align*}
b &= N\varphi_\ell\Bigl(1+\varphi_1+\varphi_1\varphi_2+\cdots+
      \prod_{\substack{i=1\\ i\ne\ell}}^{m+1}\varphi_i\Bigr).
\end{align*}
Both of these are sums of $m+1$ terms, where the $j$-th term is the
cardinality of the cluster after $j$ absorptions; call these terms $a_j$ and
$b_j$. After $j$ absorptions the first order has taken $\{u_1,\dots,u_j\}$,
while the second has taken $\{u_\ell\}$ together with $\{u_1,\dots,u_{j-1}\}$.
These two sets are the same once $j\ge\ell$, so $b_j=a_j$ for every such $j$.
If we subtract, only the first $\ell-1$ terms survive:
\begin{align*}
b-a &= N\Bigl[(\varphi_\ell-\varphi_1)+\varphi_1(\varphi_\ell-\varphi_2)
      +\varphi_1\varphi_2(\varphi_\ell-\varphi_3)+\cdots
      +(\varphi_\ell-\varphi_{\ell-1})\prod_{i=1}^{\ell-2}\varphi_i\Bigr]\\
    &= N\sum_{j=1}^{\ell-1}(\varphi_\ell-\varphi_j)\prod_{i<j}\varphi_i.
\end{align*}
Every fanout is positive, and $\varphi_j\le\varphi_\ell$ whenever $j<\ell$, so
every summand here is nonnegative and $b\ge a$. Therefore absorbing $u_1$
first is at least as cheap as any other choice, and together with the
inductive hypothesis for the rest, the ascending order is optimal.
\end{proof}

\begin{notn}
For a hub weight $N$ and a multiset $M$ of spokes we write
$\mathrm{chain}(N;M)$ for the minimum in Lemma~\ref{lem:rank} and
$\Lambda(N;M)=N\prod_{u\in M}\varphi_u$ for the resulting cardinality.
\end{notn}

\begin{lemma}[Monotonicity]\label{lem:mono}
$\mathrm{chain}(N;M)$ is nondecreasing in $N$ and in each fanout
$\varphi_u$, $u\in M$.
\end{lemma}

\begin{proof}
Fix a permutation $\pi$. The expression in Lemma~\ref{lem:rank} is then a sum
of products of $N$ and the $\varphi$'s with positive coefficients, so it is
nondecreasing in each of them. Now, $\mathrm{chain}$ is the minimum of these
expressions over the finitely many $\pi$, and a pointwise minimum of
nondecreasing functions is nondecreasing.
\end{proof}

\subsection{The shape of a double-star plan}\label{sec:shape}

In Section~\ref{sec:enum-dstar} we described the four phases of a double-star
join order informally, with pictures. We now state them as a lemma, since the
plan space that our algorithm searches is exactly the one this lemma gives us.

\begin{lemma}[Shape]\label{lem:shape}
Let $T$ be a merge tree for $S_{a,b}$. Then:
\begin{enumerate}
\item[(i)] every cluster of $T$ with at least two vertices contains $H_1$ or
$H_2$;
\item[(ii)] there is a unique internal node $x^\ast$, the \emph{big merge},
whose cluster is the smallest one containing both hubs; its two children
clusters are $C_1\ni H_1$ and $C_2\ni H_2$, and exactly one of them contains
$S$;
\item[(iii)] $C_1\setminus\{H_1,S\}\subseteq\mcA$ and
$C_2\setminus\{H_2,S\}\subseteq\mcB$, and each $C_i$ is built by its hub
absorbing singletons one at a time;
\item[(iv)] above $x^\ast$, the remaining spokes are absorbed into the growing
cluster one at a time.
\end{enumerate}
\end{lemma}

\begin{proof}
Every vertex other than $H_1$ and $H_2$ is a spoke. A vertex of $\mcA$ has
$H_1$ as its only neighbour, a vertex of $\mcB$ has $H_2$, and $S$ has $H_1$
and $H_2$. So no two vertices of $V\setminus\{H_1,H_2\}$ are adjacent, and a
connected set with two or more vertices has to contain two adjacent vertices,
and hence a hub. This gives (i).

Now, let $x^\ast$ be the node closest to the leaves whose cluster contains
both hubs. Such a node exists, since the root cluster contains both hubs, and
it is unique, since the nodes whose clusters contain both hubs form a path up
to the root. Its children clusters $C_1$ and $C_2$ partition $C(x^\ast)$, and
they cannot both contain the same hub, so one of them contains $H_1$ and the
other contains $H_2$. By Definition~\ref{def:mergetree} they are adjacent.
Now, the only edges of $S_{a,b}$ joining a set containing $H_1$ to a set
containing $H_2$ are $H_1S$ and $SH_2$, so $S$ must lie in $C_1$ or in $C_2$;
the crossing edge is $SH_2$ when $S\in C_1$ and $H_1S$ when $S\in C_2$. This
gives (ii), and these two possibilities are exactly the two cases of
Proposition~\ref{prop:dstar-sum}.

For (iii), a vertex of $\mcB$ can lie in a connected set with $H_1$ only if
that set also contains $H_2$, which $C_1$ does not, so
$C_1\subseteq\{H_1,S\}\cup\mcA$. Now, look inside the subtree at $C_1$. Each
internal node there has one child containing $H_1$ and one child not
containing it, and by (i) the latter is a singleton. The same argument gives
$C_2$.

For (iv), every node above $x^\ast$ has one child whose cluster contains
$C(x^\ast)$ and hence both hubs, and by (i) its other child is a singleton.
\end{proof}

\begin{prop}[Cost decomposition]\label{prop:decomp}
Fix which hub takes $S$, say $S\in C_1$, together with the sets
$E_1=C_1\cap\mcA$ and $E_2=C_2\cap\mcB$ of spokes absorbed before the big
merge, and let $D=(\mcA\setminus E_1)\cup(\mcB\setminus E_2)$ be the deferred
spokes. The least cost of a merge tree with these data is
\begin{align}\label{eq:decomp}
\mathrm{chain}\bigl(w_{H_1};E_1\cup\{S\}\bigr)
+\mathrm{chain}\bigl(w_{H_2};E_2\bigr)
+ N\cdot W(D),
\end{align}
where, writing $\Lambda_1=\Lambda(w_{H_1};E_1\cup\{S\})$ and
$\Lambda_2=\Lambda(w_{H_2};E_2)$,
\begin{align*}
N=\Lambda_1\Lambda_2\,\sigma_{SH_2},
\qquad
W(D)=\sum_{k=0}^{|D|}\ \prod_{i\le k}\varphi_{(i)},
\end{align*}
the $\varphi_{(i)}$ being the fanouts of $D$ in ascending order.
\end{prop}

\begin{proof}
By Lemma~\ref{lem:shape} the internal nodes of $T$ fall into four groups:
those below $x^\ast$ on the $H_1$ side, those below it on the $H_2$ side,
$x^\ast$ itself, and those above it.

The first two groups are chains of singleton absorptions, so by
Lemma~\ref{lem:rank} their least total costs are the two $\mathrm{chain}$
terms, and their cardinalities are $\Lambda_1$ and $\Lambda_2$ whatever order
we use. Now, the cluster at $x^\ast$ is $C_1\cup C_2$, whose internal edges
are those of $C_1$, those of $C_2$, and the crossing edge $SH_2$; therefore
$\card(C(x^\ast))=\Lambda_1\Lambda_2\sigma_{SH_2}=N$. Above $x^\ast$ the spokes of
$D$ are absorbed one at a time, which contributes $N$ times the sum of the
prefix products of their fanouts, and by Lemma~\ref{lem:rank} this is least in
ascending order; adding $\card(C(x^\ast))=N$ itself supplies the $k=0$ term of
$W(D)$.

Finally, the three groups do not interact with one another, so we are allowed
to minimize them separately.
\end{proof}

\subsection{Which spokes to defer}\label{sec:defer}

Proposition~\ref{prop:decomp} still leaves us with a search over the pairs of
subsets $(E_1,E_2)$, and there are $2^{a+b}$ of those. The next two lemmas
bring that down to $(a+1)(b+1)$. The first one rewrites the last term of
\eqref{eq:decomp} in a form where we can see how it depends on $D$.

\begin{lemma}[Final-size identity]\label{lem:finalsize}
Let $\Pi_t(D)$ denote the product of the $t$ largest fanouts in $D$, with
$\Pi_0(D)=1$. Then
\begin{align*}
N\cdot W(D)\;=\;F\sum_{t=0}^{|D|}\frac{1}{\Pi_t(D)},
\end{align*}
where $F$ is the instance constant of Remark~\ref{rem:F}.
\end{lemma}

\begin{proof}
Absorbing all of $D$ turns the cluster of cardinality $N$ into the root
cluster, so $F=N\prod_{u\in D}\varphi_u$. Now, list $D$ in ascending order.
The $k$-th term of $N\,W(D)$ is
\begin{align*}
N\prod_{i\le k}\varphi_{(i)}=\frac{F}{\prod_{i>k}\varphi_{(i)}}
=\frac{F}{\Pi_{|D|-k}(D)},
\end{align*}
since the fanouts after position $k$ are the $|D|-k$ largest ones.
Substituting $t=|D|-k$ gives the claim.
\end{proof}

\begin{lemma}[Interchange]\label{lem:interchange}
For each choice of which hub takes $S$, there is an optimal plan in which
$E_1$ consists of the $|E_1|$ smallest-fanout spokes of $\mcA$ and $E_2$ of
the $|E_2|$ smallest-fanout spokes of $\mcB$.
\end{lemma}

\begin{proof}
Say that $u\in E_1$ and $v\in\mcA\setminus E_1$ with $\varphi_v<\varphi_u$,
and exchange them: let $E_1'=(E_1\setminus\{u\})\cup\{v\}$ and
$D'=(D\setminus\{v\})\cup\{u\}$. We compare the three terms of
\eqref{eq:decomp}, using the form that Lemma~\ref{lem:finalsize} gives us for
the third one.

The second term does not change at all. The first term is
$\mathrm{chain}(w_{H_1};E_1\cup\{S\})$, and passing from $E_1$ to $E_1'$
replaces the fanout $\varphi_u$ by the smaller $\varphi_v$ and leaves
everything else alone, so by Lemma~\ref{lem:mono} it does not increase. Now,
for the third term, $D'$ comes from $D$ by replacing $v$ with the larger $u$.
If we sort both multisets in descending order, every order statistic of $D'$
is at least the corresponding one of $D$, so $\Pi_t(D')\ge\Pi_t(D)$ for every
$t$, and therefore $F\sum_t 1/\Pi_t(D')\le F\sum_t 1/\Pi_t(D)$. Since $F$ is
the same for every plan, the total does not increase.

Repeating the exchange terminates, since each one strictly decreases the sum
of the ranks of the spokes in $E_1$, and when no exchange is left to make,
$E_1$ is a set of smallest-fanout spokes. The argument for $E_2$ is the same,
with the first and second terms of \eqref{eq:decomp} swapped.
\end{proof}

\begin{cor}\label{cor:planspace}
An optimal merge tree for $S_{a,b}$ is attained among $2(a+1)(b+1)$ plans,
indexed by which hub takes $S$ together with the number $p$ of $\mcA$-spokes
and the number $q$ of $\mcB$-spokes absorbed before the big merge.
\end{cor}

\subsection{The algorithm}\label{sec:dstar-alg}

By Corollary~\ref{cor:planspace}, all our algorithm has to do is evaluate
\eqref{eq:decomp} for each $(p,q)$. If we do this directly then each
evaluation costs $O\bigl((a+b)\log(a+b)\bigr)$, since we have to pool the
deferred spokes of the two sides and sort them. The next lemma lets us avoid
that, by computing all $(a+1)(b+1)$ values of $W$ at once.

\begin{lemma}[Leftover recurrence]\label{lem:Wrec}
Sort $\mcA$ and $\mcB$ ascending by fanout, as
$\varphi^{\mcA}_1\le\cdots\le\varphi^{\mcA}_a$ and
$\varphi^{\mcB}_1\le\cdots\le\varphi^{\mcB}_b$, and write $W[p][q]$ for
$W(D)$ with $D$ the last $a-p$ spokes of $\mcA$ together with the last $b-q$
of $\mcB$. Then $W[a][b]=1$ and, otherwise,
\begin{align*}
W[p][q]=
\begin{cases}
1+\varphi^{\mcA}_{p+1}\,W[p+1][q], & q=b,\ \text{or } p<a \text{ and } \varphi^{\mcA}_{p+1}\le\varphi^{\mcB}_{q+1},\\[2pt]
1+\varphi^{\mcB}_{q+1}\,W[p][q+1], & \text{otherwise.}
\end{cases}
\end{align*}
\end{lemma}

\begin{proof}
$D$ is the union of a suffix of $\mcA$ and a suffix of $\mcB$, so its smallest
fanout is $\min(\varphi^{\mcA}_{p+1},\varphi^{\mcB}_{q+1})$, and deleting that
element leaves a set of the same form, with the corresponding index advanced
by one. Now, write $D=\{d_1\le\cdots\le d_m\}$ and factor $\varphi_{d_1}$ out
of every term of $W(D)$ except the $k=0$ one; this gives
$W(D)=1+\varphi_{d_1}W(D\setminus\{d_1\})$.
\end{proof}

Algorithms~\ref{alg:prefix} through \ref{alg:main} put the pieces together.
Algorithm~\ref{alg:prefix} tabulates the cost and the cardinality of the
cluster that a hub has built after absorbing its $p$ cheapest spokes. We need
two versions of this: one without the shared spoke, for the hub on the far
side of the big merge, and one with it, for the hub that takes $S$. In the
second version the shared spoke is present in every prefix, but it sits at its
own rank position, which is why the table splits at the index $j$ where
$\varphi_S$ falls.

\begin{algorithm}[htbp]
\caption{Prefix tables for one hub}\label{alg:prefix}
\begin{algorithmic}[1]
\Procedure{PrefixTables}{$w$, $\varphi_1\le\cdots\le\varphi_k$}
  \State $N[0]\gets w$; \ $C[0]\gets 0$
  \For{$p=1$ \textbf{to} $k$}
    \State $N[p]\gets N[p-1]\cdot\varphi_p$; \ $C[p]\gets C[p-1]+N[p]$
  \EndFor
  \State \Return $C,N$ \Comment{cost and cardinality after the $p$ cheapest spokes}
\EndProcedure
\Statex
\Procedure{AnchorTables}{$w$, $\varphi_1\le\cdots\le\varphi_k$, $\varphi_S$}
  \State $C,N\gets\Call{PrefixTables}{w,\varphi_1,\dots,\varphi_k}$
  \State $j\gets |\{\,i:\varphi_i<\varphi_S\,\}|$ \Comment{rank position of the shared spoke}
  \For{$p=0$ \textbf{to} $k$}
    \State $\mathit{size}[p]\gets N[p]\cdot\varphi_S$
    \If{$p\le j$}
      \State $\mathit{cost}[p]\gets C[p]+N[p]\cdot\varphi_S$
    \Else
      \State $\mathit{cost}[p]\gets C[j]+\varphi_S\cdot\bigl(N[j]+C[p]-C[j]\bigr)$
    \EndIf
  \EndFor
  \State \Return $\mathit{cost},\mathit{size}$
\EndProcedure
\end{algorithmic}
\end{algorithm}

\begin{algorithm}[htbp]
\caption{Pooled leftovers, all $(p,q)$ at once}\label{alg:leftover}
\begin{algorithmic}[1]
\Procedure{LeftoverTable}{$\varphi^{\mcA}_1\le\cdots\le\varphi^{\mcA}_a$, $\varphi^{\mcB}_1\le\cdots\le\varphi^{\mcB}_b$}
  \State $W[a][b]\gets 1$
  \For{$p=a$ \textbf{down to} $0$}
    \For{$q=b$ \textbf{down to} $0$}
      \State \textbf{if} $p=a$ \textbf{and} $q=b$ \textbf{then continue}
      \If{$q=b$ \textbf{or} ($p<a$ \textbf{and} $\varphi^{\mcA}_{p+1}\le\varphi^{\mcB}_{q+1}$)}
        \State $W[p][q]\gets 1+\varphi^{\mcA}_{p+1}\cdot W[p+1][q]$
      \Else
        \State $W[p][q]\gets 1+\varphi^{\mcB}_{q+1}\cdot W[p][q+1]$
      \EndIf
    \EndFor
  \EndFor
  \State \Return $W$
\EndProcedure
\end{algorithmic}
\end{algorithm}

\begin{algorithm}[htbp]
\caption{Optimal join order for a double star}\label{alg:main}
\begin{algorithmic}[1]
\Procedure{SolveOneSide}{$w_h$, $w_o$, $\varphi_S$, $\sigma$, $\mcH$, $\mcO$}
  \Comment{$\mcH$ is the hub that takes $S$; $\sigma$ is the far selectivity}
  \State sort $\mcH$ and $\mcO$ ascending by fanout; let $a=|\mcH|$, $b=|\mcO|$
  \State $\mathit{cost}_L,\mathit{size}_L\gets\Call{AnchorTables}{w_h,\ \varphi(\mcH),\ \varphi_S}$
  \State $\mathit{cost}_R,\mathit{size}_R\gets\Call{PrefixTables}{w_o,\ \varphi(\mcO)}$
  \State $W\gets\Call{LeftoverTable}{\varphi(\mcH),\varphi(\mcO)}$
  \State $\mathit{best}\gets\infty$
  \For{$p=0$ \textbf{to} $a$}
    \For{$q=0$ \textbf{to} $b$}
      \State $N\gets \mathit{size}_L[p]\cdot \mathit{size}_R[q]\cdot\sigma$
        \Comment{cardinality at the big merge}
      \State $t\gets \mathit{cost}_L[p]+\mathit{cost}_R[q]+N\cdot W[p][q]$
      \State \textbf{if} $t<\mathit{best}$ \textbf{then} $(\mathit{best},p^\ast,q^\ast)\gets(t,p,q)$
    \EndFor
  \EndFor
  \State \Return $\mathit{best},p^\ast,q^\ast$
\EndProcedure
\Statex
\Procedure{OptimalDoubleStar}{$w_{H_1},w_{H_2},w_S,\sigma_{H_1S},\sigma_{SH_2},\mcA,\mcB$}
  \State $r_1\gets\Call{SolveOneSide}{w_{H_1},w_{H_2},\,w_S\sigma_{H_1S},\,\sigma_{SH_2},\,\mcA,\mcB}$
  \State $r_2\gets\Call{SolveOneSide}{w_{H_2},w_{H_1},\,w_S\sigma_{SH_2},\,\sigma_{H_1S},\,\mcB,\mcA}$
  \State \Return the cheaper of $r_1$ and $r_2$, tagged with which hub took $S$
\EndProcedure
\end{algorithmic}
\end{algorithm}

\begin{thm}\label{thm:dstar-alg}
\textup{\textsc{OptimalDoubleStar}} returns
$\min_{T\in\mcT(S_{a,b})}\mathrm{cost}(T)$. It runs in time
$O\bigl(a\log a+b\log b+ab\bigr)$ and uses $O(ab)$ space.
\end{thm}

\begin{proof}
Correctness follows from Proposition~\ref{prop:decomp} together with
Corollary~\ref{cor:planspace}. The two calls to \textsc{SolveOneSide} cover
the two choices of which hub takes $S$, and inside each call the loop
evaluates \eqref{eq:decomp} for every $(p,q)$; by Lemma~\ref{lem:interchange}
this includes an optimal plan. The loop body evaluates \eqref{eq:decomp}
correctly by Lemma~\ref{lem:rank} for $\mathit{cost}_L$ and
$\mathit{cost}_R$, and by Lemma~\ref{lem:Wrec} for $W$.

Now, for the running time. The two sorts cost $O(a\log a+b\log b)$.
\textsc{PrefixTables} and \textsc{AnchorTables} are both $O(a+b)$; for the
second one this is because we can find the split index $j$ in a single pass
over an already sorted list. \textsc{LeftoverTable} fills $(a+1)(b+1)$ entries
at $O(1)$ each, and the main loop does $O(1)$ work per pair. The only thing we
store that is asymptotically large is $W$, and we can bring that down to
$O(\min(a,b))$ by generating it one row at a time as the loop runs.
\end{proof}

\begin{rem}[Recovering the plan]\label{rem:recover}
Note that the triple (which hub took $S$, $p^\ast$, $q^\ast$) determines an
optimal merge tree, and not just its cost. By Lemma~\ref{lem:interchange} the
hub that takes $S$ absorbs the $p^\ast$ cheapest spokes on its own side; by
Lemma~\ref{lem:rank} it absorbs them together with $S$ in ascending fanout
order; the far hub does the same with its $q^\ast$ cheapest spokes; and after
the big merge, the two leftover piles are pooled and absorbed in ascending
fanout order.
\end{rem}

\begin{ex}\label{ex:dstar}
Take $w_{H_1}=10$, $w_{H_2}=50$, $w_S=10$, $\sigma_{H_1S}=0.5$ and
$\sigma_{SH_2}=0.2$, with
\begin{align*}
\mcA:\quad &\varphi_{a_1}=0.1,\quad \varphi_{a_2}=0.4,\quad \varphi_{a_3}=2.5,\\
\mcB:\quad &\varphi_{b_1}=0.2,\quad \varphi_{b_2}=2.0,
\end{align*}
so that $\varphi^{(1)}_S=5$ and $\varphi^{(2)}_S=2$, and $F=20$. The two calls
to \textsc{SolveOneSide} produce the tables
\[
\begin{array}{c|ccc}
\multicolumn{4}{c}{S \text{ taken by } H_1}\\
p\backslash q & 0 & 1 & 2\\\hline
0 & 642.0 & 202.0 & 328.0\\
1 & 98.0  & 58.0  & 84.0\\
2 & 55.4  & \mathbf{45.4} & 61.4\\
3 & 87.4  & 47.4  & 57.4
\end{array}
\qquad
\begin{array}{c|cccc}
\multicolumn{5}{c}{S \text{ taken by } H_2}\\
p\backslash q & 0 & 1 & 2 & 3\\\hline
0 & 692.0 & 193.0 & 153.4 & 182.4\\
1 & 172.0 & 73.0  & 63.4  & \mathbf{62.4}\\
2 & 318.0 & 119.0 & 99.4  & 92.4
\end{array}
\]
in which $p$ always counts the spokes of the hub that takes $S$. The optimum
is $45.4$, and it occurs when $H_1$ takes $S$ with $p^\ast=2$ and $q^\ast=1$.
By Remark~\ref{rem:recover} the plan is
\begin{align*}
\Bigl(\bigl(((H_1\bowtie a_1)\bowtie a_2)\bowtie S\bigr)\bowtie (H_2\bowtie b_1)\Bigr)
\bowtie b_2\ \bowtie\ a_3,
\end{align*}
whose intermediate cardinalities are $1,\,0.4,\,2,\,10,\,4,\,8,\,20$. Note
that the optimum defers $a_3$ and $b_2$, which are the two expanding spokes.
Note also that the two orientations differ by more than a quarter of the total
cost, so we can drop neither the search over deferrals nor the second call to
\textsc{SolveOneSide}.
\end{ex}

\subsection{The helix}\label{sec:helix-alg}
 
For the double star we were able to cut the plan space down to $2(a+1)(b+1)$
candidates, because Lemma~\ref{lem:shape} told us exactly what a plan has to
look like. For the helix we have no such lemma. Theorem~\ref{thm:hn} tells us
how many merge trees $\mcH_n$ has, but it says nothing about which of them is
the cheapest, and the counting argument behind it does not survive the
introduction of weights: two merge trees counted by the same $E_n(d)$ can have
very different costs. So instead we go back to Lemma~\ref{lem:split} and turn
it into a dynamic program. Similar observations were made by~\cite{gpu_joe}.
 
The reason this works is that the cost of a merge tree splits at its root.
 
\begin{lemma}[Bellman equation]\label{lem:bellman}
For a connected $U\subseteq V$, write
$\mu(U)=\min\{\mathrm{cost}(T):T\in\mcT(G[U])\}$. Then $\mu(U)=0$ when
$|U|=1$, and otherwise
\begin{align*}
\mu(U)=\card(U)+\min_{\substack{\emptyset\ne L\subsetneq U\\ L,\,U\setminus L \text{ connected}}}
\bigl(\mu(L)+\mu(U\setminus L)\bigr).
\end{align*}
\end{lemma}
 
\begin{proof}
Say $|U|\ge 2$ and let $T$ be a merge tree for $G[U]$. Its root has two
children, whose clusters $L$ and $U\setminus L$ are connected and partition
$U$, and the two subtrees hanging off the root are merge trees for $G[L]$ and
$G[U\setminus L]$. Conversely, any such split together with a merge tree for
each part gives a merge tree for $G[U]$. Now, by Remark~\ref{rem:cout} the
quantity $\card(C(x))$ depends only on the cluster at $x$ and not on how that
cluster was built, so
\begin{align*}
\mathrm{cost}(T)=\card(U)+\mathrm{cost}(T_L)+\mathrm{cost}(T_R),
\end{align*}
where $T_L$ and $T_R$ are the two subtrees. Therefore minimizing over $T$ is
the same as minimizing the two subtrees separately and then minimizing over
the split, which is what the displayed equation says.
\end{proof}
 
Note that this is the same recursion as Lemma~\ref{lem:split}, over the same
splits; we have simply replaced the sum of products by a minimum of sums.
There are two small things to check before we can compute with it.
 
\begin{lemma}\label{lem:splitexists}
Let $U$ be connected with $|U|\ge 2$ and let $z\in U$. Then $U$ has at least
one split into two connected parts, and we may take $z$ to lie in whichever
part we like.
\end{lemma}
 
\begin{proof}
Take a spanning tree of $G[U]$. It has at least two leaves, so we may pick a
leaf $v$ with $v\ne z$. Deleting a leaf from a tree leaves a tree, so
$U\setminus\{v\}$ is connected, and $\{v\}$ is connected as well; this is a
split, and $z$ lies in $U\setminus\{v\}$. Swapping the names of the two parts
gives the other case.
\end{proof}
 
So the minimum in Lemma~\ref{lem:bellman} is never over an empty set. Now, the
two parts of a split are unordered, so if we enumerate all pairs
$(L,U\setminus L)$ we look at each split twice. We avoid this exactly the way
Lemma~\ref{lem:split} does. Fix any total order on the vertices, let $z$ be
the smallest vertex of $U$, and only consider those $L$ that contain $z$;
every split has $z$ in exactly one of its two parts, so this enumerates each
split once. By Lemma~\ref{lem:splitexists} we still see at least one.
 
\begin{rem}\label{rem:crossredundant}
It is natural to also test that $L$ and $U\setminus L$ are adjacent, so that
the merge at the root is not a cross product. This test is redundant here.
Both parts are connected and they partition the connected set $U$, so by
Remark~\ref{rem:model}(ii) some edge of $G$ has to cross between them.
\end{rem}
 
Algorithm~\ref{alg:helix} is the resulting dynamic program. We fill a table
$\mu$ over the connected subsets, working upwards by size, so that whenever we
come to $U$ both parts of every split of $U$ are already in the table, and
alongside it a table $\mathit{split}$ recording which split achieved each
minimum. Note that $\mu$ also does the connectivity test for us: its keys are
exactly the connected sets we have already processed, so a part is connected
if and only if we can look it up.
 
\begin{algorithm}[htbp]
\caption{Optimal join order for a helix}\label{alg:helix}
\begin{algorithmic}[1]
\Procedure{HelixMinCost}{$V$, $\mathrm{adj}$, weights $w$, selectivities $\sigma$}
  \State $\mathcal{C}\gets$ the connected subsets of $V$, sorted by size
  \For{$U\in\mathcal{C}$}
    \State $\mathit{card}[U]\gets\prod_{v\in U}w_v\prod_{e\in E(G[U])}\sigma_e$
      \Comment{the estimated cardinality $\card(U)$}
  \EndFor
  \For{$U\in\mathcal{C}$, in nondecreasing size}
    \If{$|U|=1$}
      \State $\mu[U]\gets 0$; \textbf{continue}
    \EndIf
    \State $\mathit{best}\gets\infty$; \ $z\gets\min U$
      \Comment{the anchor; any fixed total order will do}
    \For{each proper nonempty $L\subset U$ with $z\in L$}
      \State $R\gets U\setminus L$
        \If{$L,R$ are both keys of $\mu$ \textbf{and} $\mu[L]+\mu[R]<\mathit{best}$}
            \Comment{both connected}
            \State $\mathit{best}\gets\mu[L]+\mu[R]$; \ $\mathit{split}[U]\gets(L,R)$
        \EndIf
    \EndFor
    \State $\mu[U]\gets\mathit{best}+\mathit{card}[U]$
  \EndFor
  \State \Return $\mu[V]$ \textbf{and} $\mathit{split}$
\EndProcedure
\end{algorithmic}
\end{algorithm}
 
The number of table entries is the number of connected subsets of the helix,
which we can count exactly.
 
\begin{prop}\label{prop:connsubsets}
For $n\ge 1$, the number of nonempty connected subsets of $\mcH_n$ is
\begin{align*}
3^{\,n+2}-6n-8.
\end{align*}
\end{prop}
 
\begin{proof}
Let $U$ be such a subset. We split into two cases according to whether $U$
contains a pinch point.
 
Suppose first that it does not, so that $U$ consists only of vertices $A_i$
and $B_i$. No two of these are adjacent, since the neighbours of $A_i$ and of
$B_i$ are all pinch points, so $U$ has to be a single vertex. This gives $2n$
subsets.
 
Now suppose $U$ does contain a pinch point, and let $s$ and $t$ be the
smallest and largest indices with $P_s,P_t\in U$. Every $P_i$ with $s<i<t$ is
a cut vertex separating $P_s$ from $P_t$, so all of them lie in $U$ as well,
and the pinch points of $U$ are exactly $P_s,\dots,P_t$. For each $i$ with
$s<i\le t$ the set $U$ must contain $A_i$ or $B_i$, since otherwise $P_{i-1}$
and $P_i$ would be in different components; and any of the three nonempty
choices keeps $U$ connected. That is $3^{\,t-s}$ possibilities. Outside this
range there are at most two more diamonds to consider, namely the one below
$P_s$ and the one above $P_t$. Their middle vertices are adjacent to $P_s$
and to $P_t$ respectively, so including any of them keeps $U$ connected and
excluding them all does too; each of these two diamonds therefore contributes
a free factor of $4$ when it exists, that is, when $s\ge 1$ and when $t\le
n-1$. Altogether the count is
\begin{align*}
\sum_{0\le s\le t\le n} 3^{\,t-s}\cdot\bigl[s\ge 1\bigr]_4\cdot\bigl[t\le n-1\bigr]_4,
\end{align*}
where $[\,\cdot\,]_4$ denotes $4$ when the condition holds and $1$ when it
does not. Splitting this sum into its four cases,
\begin{align*}
s=0,\ t=n:&\quad 3^{\,n},\\
s=0,\ t<n:&\quad 4\sum_{t=0}^{n-1}3^{\,t}=2\left(3^{\,n}-1\right),\\
s>0,\ t=n:&\quad 4\sum_{s=1}^{n}3^{\,n-s}=2\left(3^{\,n}-1\right),\\
s>0,\ t<n:&\quad 16\sum_{d=0}^{n-2}(n-1-d)\,3^{\,d}=4\left(3^{\,n}-2n-1\right),
\end{align*}
and adding these gives $9\cdot 3^{\,n}-8n-8$. Adding the $2n$ subsets from the
first case gives $3^{\,n+2}-6n-8$.
\end{proof}
 
\begin{thm}\label{thm:helix-alg}
\textup{\textsc{HelixMinCost}} returns $\min_{T\in\mcT(G)}\mathrm{cost}(T)$
for any weighted join graph $G$. On $\mcH_n$ it fills $3^{\,n+2}-6n-8$ table
entries, and the split enumeration examines
\begin{align*}
\sum_{\substack{U \text{ connected}\\ |U|\ge 2}}\bigl(2^{\,|U|-1}-1\bigr)
\;=\;\Theta\bigl(16^{\,n}\bigr)
\end{align*}
candidate pairs, so the running time is $\Theta(16^{\,n})$ up to polynomial
factors. Since $\mcH_n$ has $3n+1$ vertices, this is about $2.52$ to the power
of the number of relations.
\end{thm}
 
\begin{proof}
Correctness is Lemma~\ref{lem:bellman}. The table is filled in nondecreasing
order of size, so when we reach $U$ every connected proper subset of $U$
already has its entry, and the inner loop therefore ranges over exactly the
splits appearing in the lemma, each one once by the anchor argument above; by
Lemma~\ref{lem:splitexists} at least one such split exists, so
$\mathit{best}$ is finite. The count of table entries is
Proposition~\ref{prop:connsubsets}. For the enumeration, a connected set $U$
has $2^{\,|U|-1}-1$ proper nonempty subsets containing its anchor, and
summing $2^{\,|U|}$ over the connected subsets of $\mcH_n$ gives
$\bigl(128\cdot 16^{\,n}-170n-78\bigr)/25$, from which the stated order of
growth follows.
\end{proof}

\begin{rem}[Recovering the plan]\label{rem:helix-recover}
As with the double star we want the plan itself and not only its cost, which is
why Algorithm~\ref{alg:helix} records the split that achieved each minimum. The
merge tree is then read off by recursion: the root splits $V$ into
$\mathit{split}[V]=(L,R)$, the subtree over $L$ is built the same way from
$\mathit{split}[L]$, likewise for $R$, and the singletons are the leaves. This
consults one table entry per internal node of the tree, so it costs $3n$
lookups on top of the search and nothing in the way of extra asymptotics.
\end{rem}
 
\begin{rem}[What the algorithm costs against what it saves]\label{rem:helixgap}
The dynamic program is enormously better than working through the merge trees
one at a time, but it is still exponential, and as written it looks at many
more candidate splits than there are real ones:
\[
\begin{array}{c|r|r|r|r}
n & f(\mcH_n) & \text{table entries} & \text{splits} & \text{pairs examined}\\\hline
1 & 10 & 13 & 18 & 23\\
2 & 856 & 61 & 204 & 586\\
3 & 131{,}072 & 217 & 1{,}254 & 10{,}257\\
4 & 28{,}260{,}704 & 697 & 5{,}976 & 167{,}060\\
5 & 7{,}790{,}029{,}440 & 2{,}149 & 24{,}954 & 2{,}682{,}187\\
6 & 2{,}607{,}019{,}229{,}184 & 6{,}517 & 96{,}420 & 42{,}943{,}134
\end{array}
\]
The last two columns differ because we enumerate all subsets of $U$ that
contain the anchor and then discard the ones whose two parts are not both
connected; enumerating only the splits themselves would bring the running time
down to the fourth column. Note what the helix is doing for us in this table.
It is not making the algorithm any faster. It is making the algorithm's cost
computable: Proposition~\ref{prop:connsubsets} and Theorem~\ref{thm:helix-alg}
are exact statements about a dynamic program that knows nothing about helices,
and for a general family we could only have given bounds. We note that our observations mirror those in~\cite{dpcpp}, where the authors make a similar observation about testing candidate splits and enumerating valid connected splits.
\end{rem}
 
\begin{ex}\label{ex:helix}
Take $\mcH_1$, whose vertices are $P_0$, $P_1$, $A_1$ and $B_1$, with
\begin{align*}
w_{P_0}=5,\quad w_{A_1}=4,\quad w_{B_1}=20,\quad w_{P_1}=20,
\end{align*}
and
\begin{align*}
\sigma_{P_0A_1}=2,\quad \sigma_{P_0B_1}=0.05,\quad
\sigma_{A_1P_1}=0.05,\quad \sigma_{B_1P_1}=0.5.
\end{align*}
By Proposition~\ref{prop:connsubsets} there are $27-6-8=13$ connected subsets,
namely the four singletons, the four edges, the four triples and the whole
graph. The table the algorithm fills is
\[
\begin{array}{l|rr}
U & \card(U) & \mu(U)\\\hline
\{A_1,P_1\}          & 4   & 4\\
\{B_1,P_0\}          & 5   & 5\\
\{A_1,P_0\}          & 40  & 40\\
\{B_1,P_1\}          & 200 & 200\\
\{A_1,B_1,P_1\}      & 40  & 44\\
\{A_1,P_0,P_1\}      & 40  & 44\\
\{A_1,B_1,P_0\}      & 40  & 45\\
\{B_1,P_0,P_1\}      & 50  & 55\\
\{A_1,B_1,P_0,P_1\}  & 20  & 29
\end{array}
\]
so the optimal cost is $29$. Note that the split achieving it is
$\{A_1,P_1\}$ against $\{B_1,P_0\}$, giving the plan
\begin{align*}
(A_1\bowtie P_1)\bowtie(B_1\bowtie P_0)
\end{align*}
with intermediate cardinalities $4$, $5$ and $20$. This plan is genuinely
bushy - and the bushiness here is neccesary; every split of the whole graph that puts a single vertex on one side
costs at least $44+20=64$ here, so a left-deep plan is more than twice as
expensive as the optimum. This is a real difference from the double star,
where Lemma~\ref{lem:shape} tells us that no plan is ever bushy in this way.
\end{ex}
 
\begin{rem}[Where the structure went]\label{rem:helix-general}
It is worth being exact about what this subsection has and has not done.
Algorithm~\ref{alg:helix} uses no property of the helix whatsoever; it is the
generic connected-subset dynamic program and it runs on any join graph, the
double star included. The structure of Theorem~\ref{thm:K} went
into the analysis, but it did not go into the algorithm. 

Section~\ref{sec:dstar-alg} is the other way round. There Lemma~\ref{lem:shape}
changes what the algorithm does --- it is why we search $2(a+1)(b+1)$
candidates rather than run a dynamic program over subsets --- and the payoff is
a polynomial running time. So we have one positive case and one negative case,
and what separates them is exactly whether the structure reached the algorithm
or stopped at the analysis. Whether the structure behind Theorem~\ref{thm:K}
can be made to do algorithmic work as well remains open.
\end{rem}

\section{Benchmarking Against The General Algorithm}\label{sec:bench}

We now measure the two algorithms of Section~\ref{sec:algorithms} against
DPhyp~\cite{dphyp}, the general-purpose join order optimizer that a
specialized algorithm has to beat if it is to be worth having. DPhyp is exact,
so whenever our algorithm and DPhyp report the same plan cost we know that our
algorithm has found an optimal plan.
 
For each family we generate instances of increasing size and plan each one
$100$ times, recording the time spent choosing the join order and the cost of
the plan chosen. For the double star we take every size from $5$ to $21$
relations; for the helix we take $\mcH_1$ through $\mcH_8$, that is $4$, $7$,
\ldots, $25$ relations. We report medians; the interquartile range never
exceeds $14\%$ of the median at any point, and is under $5\%$ at every point
past the smallest few, where the measurement is dominated by fixed overhead.
 
The first thing to check is that we are comparing like with like.
 
\begin{rem}\label{rem:costsagree}
At every size in both families, our algorithm and DPhyp report exactly the same
plan cost. Since DPhyp is exact, this says that
\textsc{OptimalDoubleStar} and \textsc{HelixMinCost} return optimal plans on
every instance we measured, which is what Theorems~\ref{thm:dstar-alg}
and~\ref{thm:helix-alg} predict.
\end{rem}
 
\subsection{The double star}\label{sec:bench-dstar}
 
Figure~\ref{fig:bench-dstar} plots the two arms against each other. The
vertical axis is logarithmic, so a straight line means exponential growth.
 
\begin{figure}[htbp]
\centering
\begin{tikzpicture}
\begin{semilogyaxis}[
  width=0.74\linewidth, height=6.2cm,
  xlabel={number of relations}, ylabel={join order time (ms)},
  xmin=4.4, xmax=22.4, ymin=0.05, ymax=20000,
  xtick={5,7,9,11,13,15,17,19,21},
  ytick={0.1,1,10,100,1000,10000},
  yticklabels={$0.1$,$1$,$10$,$100$,$1000$,$10\,000$},
  grid=major, grid style={gray!25, very thin},
  tick align=outside, tick style={gray!50},
  axis line style={gray!60},
  label style={font=\small}, tick label style={font=\footnotesize},
  clip=false,
]
\addplot[black, thick, mark=square, mark size=1.7pt, dashed] coordinates {
 (5,0.1175) (6,0.1840) (7,0.2630) (8,0.4305) (9,0.7460) (10,1.4150)
 (11,2.8060) (12,5.5060) (13,11.1800) (14,22.9885) (15,47.9655)
 (16,101.4025) (17,216.8495) (18,467.8225) (19,1017.4840) (20,2219.1215)
 (21,5017.1640)};
\node[anchor=west, font=\footnotesize] at (axis cs:21.25,5017) {DPhyp};
\addplot[black, thick, mark=*, mark size=1.5pt] coordinates {
 (5,0.0970) (6,0.1150) (7,0.1410) (8,0.1485) (9,0.1670) (10,0.1900)
 (11,0.2090) (12,0.2350) (13,0.2560) (14,0.2830) (15,0.3095) (16,0.3395)
 (17,0.3610) (18,0.3970) (19,0.4260) (20,0.4620) (21,0.5000)};
\node[anchor=west, font=\footnotesize] at (axis cs:21.25,0.50) {Algorithm~\ref{alg:main}};
\end{semilogyaxis}
\end{tikzpicture}
\caption{Median join order time on double stars, over $100$ runs per point.
Solid with round markers is Algorithm~\ref{alg:main}; dashed with square
markers is DPhyp. The vertical axis is logarithmic.}
\label{fig:bench-dstar}
\end{figure}
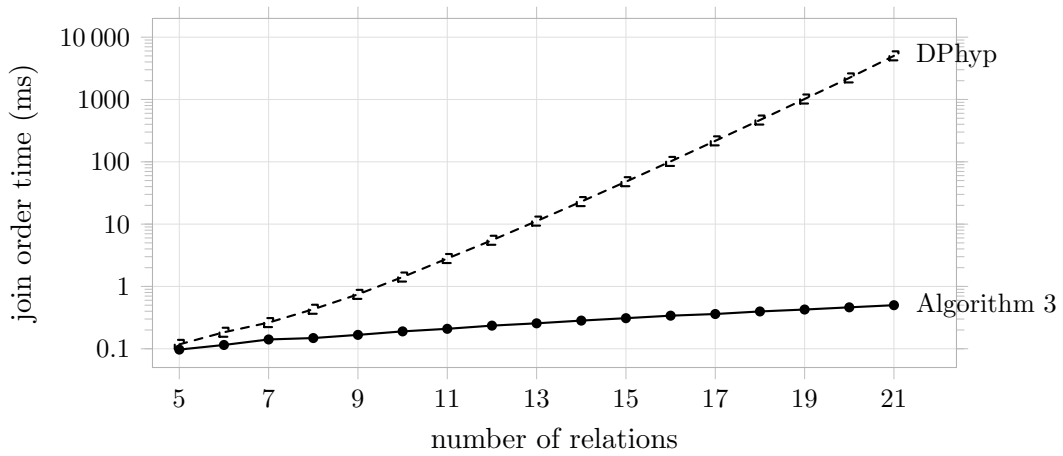
 
The two curves say exactly what Theorem~\ref{thm:dstar-alg} and
Corollary~\ref{cor:planspace} lead us to expect. DPhyp is a straight line on
the log axis: each extra relation multiplies its running time by about
$1.95$, so it is doubling roughly every relation we add. Our algorithm is
flat by comparison. In fact over this range its median time is very nearly
linear in the number of relations,
\begin{align*}
0.0247\,t-0.048 \text{ ms},\qquad R^2=0.986,
\end{align*}
which sits comfortably inside the $O(a\log a+b\log b+ab)$ bound; the quadratic
term has not had a chance to show itself at these sizes. At $5$ relations the
two are within a factor of $1.2$ of each other, which is the fixed overhead of
planning at all, but the gap opens quickly: by $13$ relations it is $44$-fold,
and at $21$ relations it is $5017$~ms against $0.50$~ms, a factor of about
$10^4$.
 
\subsection{The helix}\label{sec:bench-helix}
 
The helix tells a more interesting story, shown in
Figure~\ref{fig:bench-helix}.
 
\begin{figure}[htbp]
\centering
\begin{tikzpicture}
\begin{semilogyaxis}[
  width=0.74\linewidth, height=6.2cm,
  xlabel={number of relations}, ylabel={join order time (ms)},
  xmin=3.2, xmax=26.8, ymin=0.05, ymax=60000,
  xtick={4,7,10,13,16,19,22,25},
  ytick={0.1,1,10,100,1000,10000},
  yticklabels={$0.1$,$1$,$10$,$100$,$1000$,$10\,000$},
  grid=major, grid style={gray!25, very thin},
  tick align=outside, tick style={gray!50},
  axis line style={gray!60},
  label style={font=\small}, tick label style={font=\footnotesize},
  clip=false,
]
\addplot[black, thick, mark=square, mark size=1.7pt, dashed] coordinates {
 (4,0.1080) (7,0.4855) (10,2.3620) (13,10.6255) (16,45.8065)
 (19,187.8850) (22,747.5160) (25,2987.7450)};
\node[anchor=west, font=\footnotesize] at (axis cs:25.3,2988) {DPhyp};
\addplot[black, thick, mark=*, mark size=1.5pt] coordinates {
 (4,0.0850) (7,0.1540) (10,0.2820) (13,0.7970) (16,5.2490)
 (19,63.6035) (22,934.5890) (25,14477.4775)};
\node[anchor=west, font=\footnotesize] at (axis cs:25.3,14477) {Algorithm~\ref{alg:helix}};
\draw[gray!70, thin, dotted] (axis cs:20.6,0.05) -- (axis cs:20.6,60000);
\node[anchor=south, font=\footnotesize, gray!70!black]
  at (axis cs:20.6,0.06) {crossover};
\end{semilogyaxis}
\end{tikzpicture}
\caption{Median join order time on helices $\mcH_1$ through $\mcH_8$, over
$100$ runs per point. Solid with round markers is Algorithm~\ref{alg:helix};
dashed with square markers is DPhyp. The two curves cross between $19$ and
$22$ relations.}
\label{fig:bench-helix}
\end{figure}
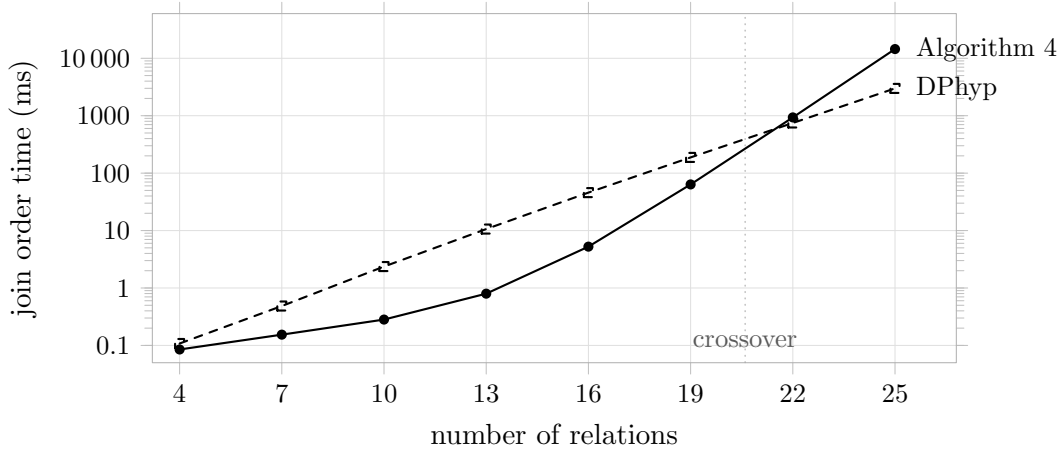
 
Up to $\mcH_4$ the dynamic program wins comfortably, peaking at about
$13$-fold at $13$ relations. After that its advantage erodes, and somewhere
between $19$ and $22$ relations DPhyp overtakes it; by $\mcH_8$ our algorithm
is nearly five times \emph{slower}. This is not a surprise once we look at how
fast each curve is climbing. Adding one diamond multiplies the measured time
of Algorithm~\ref{alg:helix} by
\begin{align*}
1.81,\quad 1.83,\quad 2.83,\quad 6.59,\quad 12.12,\quad 14.69,\quad 15.49,
\end{align*}
which is settling on $16$, exactly the constant that Theorem~\ref{thm:helix-alg}
predicts. The same figures for DPhyp are
\begin{align*}
4.50,\quad 4.87,\quad 4.50,\quad 4.31,\quad 4.10,\quad 3.98,\quad 4.00,
\end{align*}
so DPhyp is growing roughly fourfold per diamond. Both are exponential, but
one has four times the base of the other, and no constant factor is going to
rescue that.
 
\begin{rem}\label{rem:whyhelixloses}
Note that this is a race between two general-purpose algorithms. Neither
Algorithm~\ref{alg:helix} nor DPhyp knows anything about helices, and what
separates them is only how they enumerate: DPhyp goes through the connected
splits directly, whereas Algorithm~\ref{alg:helix} goes through every subset of
$U$ containing the anchor and discards the ones whose two parts are not both
connected. What the helix contributes here is the constants. The table in
Remark~\ref{rem:helixgap} counts both quantities exactly, and over that range
the splits grow by a factor of roughly $4$ per diamond while the pairs examined
grow by a factor of $16$; those are the two slopes in
Figure~\ref{fig:bench-helix}, so the crossover was predictable before we ran
anything. Replacing the subset enumeration by a direct enumeration of the
connected splits would put Algorithm~\ref{alg:helix} on DPhyp's growth rate,
and the measured constants suggest it would then stay ahead. We have not done
this. It is a limitation of the enumeration strategy rather than of
Lemma~\ref{lem:bellman}, and it is the obvious next step.
\end{rem}
 
\subsection{Speedup}\label{sec:bench-speedup}
 
Figure~\ref{fig:bench-speedup} puts both families on one axis, plotting how
many times faster our algorithm is than DPhyp. Above the line at $1$ we win;
below it we lose.
 
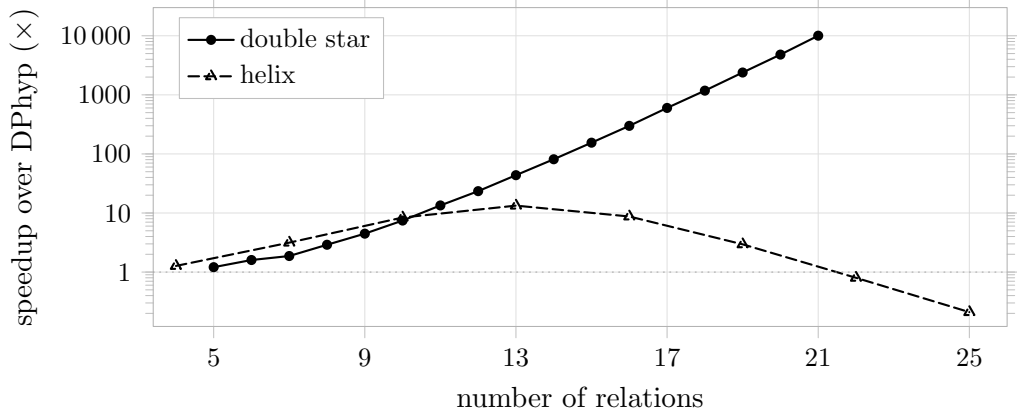
\begin{figure}[htbp]
\centering
\begin{tikzpicture}
\begin{semilogyaxis}[
  width=0.78\linewidth, height=5.8cm,
  xlabel={number of relations}, ylabel={speedup over DPhyp ($\times$)},
  xmin=3.4, xmax=26, ymin=0.12, ymax=30000,
  xtick={5,9,13,17,21,25},
  ytick={0.1,1,10,100,1000,10000},
  yticklabels={$0.1$,$1$,$10$,$100$,$1000$,$10\,000$},
  grid=major, grid style={gray!25, very thin},
  tick align=outside, tick style={gray!50},
  axis line style={gray!60},
  label style={font=\small}, tick label style={font=\footnotesize},
  legend style={at={(0.03,0.97)}, anchor=north west, draw=gray!50,
                font=\footnotesize, cells={anchor=west}, row sep=1pt},
]
\addplot[black, thick, mark=*, mark size=1.5pt] coordinates {
 (5,1.21) (6,1.60) (7,1.87) (8,2.90) (9,4.47) (10,7.45) (11,13.43)
 (12,23.43) (13,43.67) (14,81.23) (15,154.98) (16,298.68) (17,600.69)
 (18,1178.39) (19,2388.46) (20,4803.29) (21,10034.33)};
\addlegendentry{double star}
\addplot[black, thick, mark=triangle, mark size=2.2pt, dash pattern=on 4pt off 2pt] coordinates {
 (4,1.27) (7,3.15) (10,8.38) (13,13.33) (16,8.73) (19,2.95) (22,0.80) (25,0.21)};
\addlegendentry{helix}
\addplot[gray!70, thin, dotted, forget plot] coordinates {(3.4,1) (26,1)};
\end{semilogyaxis}
\end{tikzpicture}
\caption{Median speedup over DPhyp, by family. The dotted line marks parity.
The double-star curve is Algorithm~\ref{alg:main}; the helix curve is
Algorithm~\ref{alg:helix}.}
\label{fig:bench-speedup}
\end{figure}
 
The two curves are the summary of the paper. The double-star algorithm gains
about a factor of two on DPhyp with every relation we add and shows no sign of
stopping, because it is polynomial and DPhyp is not; this is the positive case,
and it is positive because Lemma~\ref{lem:shape} changed what the algorithm
does. The helix curve rises, turns over at $\mcH_4$ and comes back down through
parity; this is the negative case, and it is negative because there the
structure went into the analysis instead, as
Remark~\ref{rem:helix-general} says. Both halves are worth having. Structure
that only sharpens the analysis of a generic method tells us precisely how far
that method will get (here, that it is overtaken between $19$ and $22$
relations), but it does not move the frontier. Structure that changes what
the algorithm does is what moves it.
 
\begin{table}[htbp]
\centering
\small
\begin{tabular}{r|rrr|r|rrr}
\multicolumn{4}{c|}{double star} & & \multicolumn{3}{c}{helix}\\
$t$ & Alg.~\ref{alg:main} & DPhyp & speedup & $t$ & Alg.~\ref{alg:helix} & DPhyp & speedup\\\hline
 5 & 0.097 &    0.118 &     1.2 &  4 &     0.085 &    0.108 & 1.3\\
 7 & 0.141 &    0.263 &     1.9 &  7 &     0.154 &    0.486 & 3.2\\
 9 & 0.167 &    0.746 &     4.5 & 10 &     0.282 &    2.362 & 8.4\\
11 & 0.209 &    2.806 &    13.4 & 13 &     0.797 &   10.626 & 13.3\\
13 & 0.256 &   11.180 &    43.7 & 16 &     5.249 &   45.807 & 8.7\\
15 & 0.310 &   47.966 &   155.0 & 19 &    63.604 &  187.885 & 3.0\\
17 & 0.361 &  216.850 &   600.7 & 22 &   934.589 &  747.516 & 0.8\\
19 & 0.426 & 1017.484 &  2388.5 & 25 & 14477.478 & 2987.745 & 0.2\\
21 & 0.500 & 5017.164 & 10034.3 &    &           &          &
\end{tabular}
\caption{Median join order time in milliseconds over $100$ runs, by number of
relations $t$. The double-star columns show every other size.}
\label{tab:bench}
\end{table}

\section{Related work}

Join order selection is one of the oldest problems in data management~\cite{systemr}. This is attributable to the fact that the number of possible join orders follow the Catalan numbers, growing rapidly with the number of tables involved~\cite{joe_complexity}. Most systems today combine a cost model, a cardinality estimator, and a plan enumerator to pick join orders~\cite{howgood}, with this work focusing on plan enumeration. The complexity of the enumeration process also depends on the assumptions one can make about the cost model. Assuming only that the cost model maintains optimal substructure leads to $O(3^n)$ algorithms for general purpose graphs~\cite{dphyp}. With some additional modest assumptions, this can be pushed to $O(2^n n^3)$~\cite{dpconv}. 

Making assumptions about the join graph, or the shape of the optimal result, can significantly cut the required time~\cite{joe_complexity}. Searching for only left-deep trees on acyclic graphs with a simplified cost model has led to systems with $O(n^2)$ search times, although this excludes a large family of both queries and plans~\cite{star_schema_jo}.

\section{Conclusion}\label{sec:conclusion}
There is still more to be discovered; in particular, the helix algorithm can potentially be improved. As previously mentioned, the helix algorithm is only better in theory than the general algorithm because it has less to test, and it does not take advantage of the repetitive structure of the helix. The exploration of a better algorithm for the helix is left to the reader.

\bibliographystyle{amsplain}
\bibliography{ryan-cites-long}

\providecommand{\bysame}{\leavevmode\hbox to3em{\hrulefill}\thinspace}
\providecommand{\MR}{\relax\ifhmode\unskip\space\fi MR }
\providecommand{\MRhref}[2]{\href{http://www.ams.org/mathscinet-getitem?mr=#1}{#2}
}
\providecommand{\href}[2]{#2}
\begin{thebibliography}{10}

\bibitem{path_joe}
Altan Birler, Mihail Stoian, and Thomas Neumann, \emph{Optimizing {Linearized}
  {Join} {Enumeration} by {Adapting} to the {Query} {Structure}}, Lecture Notes
  in Informatics (LNI) (2025) (en).

\bibitem{cross_products_joe}
Sophie Cluet and Guido Moerkotte, \emph{On the complexity of generating optimal
  left-deep processing trees with cross products}, Database {Theory} — {ICDT}
  '95 (Berlin, Heidelberg) (Georg Gottlob and Moshe~Y. Vardi, eds.), {ICDT}
  '95, Springer, 1995, pp.~54--67 (en).

\bibitem{tree_jo_theory}
Toshihide Ibaraki and Tiko Kameda, \emph{On the optimal nesting order for
  computing {N}-relational joins}, ACM Transactions on Database Systems
  \textbf{9} (1984), no.~3, 482--502 (en).

\bibitem{star_schema_jo}
Ravi Krishnamurthy, Haran Boral, and Carlo Zaniolo, \emph{Optimization of
  {Nonrecursive} {Queries}}, PVLDB (1986).

\bibitem{howgood}
Viktor Leis, Andrey Gubichev, Atanas Mirchev, Peter Boncz, Alfons Kemper, and
  Thomas Neumann, \emph{How {Good} {Are} {Query} {Optimizers}, {Really}?},
  PVLDB \textbf{9} (2015), no.~3, 204--215.

\bibitem{gpu_joe}
Riccardo Mancini, Srinivas Karthik, Bikash Chandra, Vasilis Mageirakos, and
  Anastasia Ailamaki, \emph{Efficient {Massively} {Parallel} {Join}
  {Optimization} for {Large} {Queries}}, Proceedings of the 2022
  {International} {Conference} on {Management} of {Data} (New York, NY, USA),
  {SIGMOD} '22, Association for Computing Machinery, June 2022, pp.~122--135.

\bibitem{dpcpp}
Guido Moerkotte and Thomas Neumann, \emph{Analysis of two existing and one new
  dynamic programming algorithm for the generation of optimal bushy join trees
  without cross products}, Proceedings of the 32nd international conference on
  {Very} large data bases (Seoul, Korea), {VLDB} '06, VLDB Endowment, September
  2006, pp.~930--941.

\bibitem{dphyp}
\bysame, \emph{Dynamic programming strikes back}, Proceedings of the 2008 {ACM}
  {SIGMOD} international conference on {Management} of data (New York, NY,
  USA), {SIGMOD} '08, Association for Computing Machinery, June 2008,
  pp.~539--552.

\bibitem{joe_complexity}
Kiyoshi Ono and Guy~M. Lohman, \emph{Measuring the {Complexity} of {Join}
  {Enumeration} in {Query} {Optimization}}, {VLDB}, {VLDB} '90, 1990,
  pp.~314--325.

\bibitem{systemr}
P.~Griffiths Selinger, M.~M. Astrahan, D.~D. Chamberlin, R.~A. Lorie, and T.~G.
  Price, \emph{Access {Path} {Selection} in a {Relational} {Database}
  {Management} {System}}, {SIGMOD} '79 (San Francisco (CA)) (John Mylopolous
  and Michael Brodie, eds.), {SIGMOD} '79, Morgan Kaufmann, 1979, pp.~511--522.

\bibitem{dpconv}
Mihail Stoian and Andreas Kipf, \emph{{DPconv}: {Super}-{Polynomially} {Faster}
  {Join} {Ordering}}, Proceedings of the ACM on Management of Data \textbf{2}
  (2024), no.~6, 1--26 (en).

\end{thebibliography}

\end{document}